\pdfoutput=1
\RequirePackage{ifpdf}
\ifpdf 
\documentclass[pdftex]{sigma}
\else
\documentclass{sigma}
\fi

\numberwithin{equation}{section}

\newtheorem{prop}{Proposition}[section]
\newtheorem{thm}[prop]{Theorem}
\newtheorem{cor}[prop]{Corollary}
\newtheorem{Lemma}[prop]{Lemma}

\theoremstyle{definition}
\newtheorem{Definition}[prop]{Definition}

\def\gF{{\mathfrak F}}
\newcommand{\bbone}{{\bf 1}}
\newcommand{\cId}{\operatorname{Id}}
\newcommand{\cB}{{\cal{B}}}
\newcommand{\cC}{{\cal{C}}}
\newcommand{\cD}{{\cal{D}}}
\newcommand{\cL}{{\cal{L}}}
\newcommand{\cJ}{{\cal{J}}}
\newcommand{\cF}{{\cal{F}}}
\newcommand{\cT}{{\cal{T}}}
\newcommand{\cS}{{\mathcal S}}
\newcommand{\Tr}{\operatorname{Tr}}

\begin{document}


\renewcommand{\thefootnote}{$\star$}

\newcommand{\arXivNumber}{1603.07312}

\renewcommand{\PaperNumber}{085}

\FirstPageHeading

\ShortArticleName{Constructive Tensor Field Theory}

\ArticleName{Constructive Tensor Field Theory\footnote{This paper is a~contribution to the Special Issue on Tensor Models, Formalism and Applications. The full collection is available at \href{http://www.emis.de/journals/SIGMA/Tensor_Models.html}{http://www.emis.de/journals/SIGMA/Tensor\_{}Models.html}}}

\Author{Vincent RIVASSEAU}

\AuthorNameForHeading{V.~Rivasseau}

\Address{Laboratoire de Physique Th\'eorique, CNRS UMR 8627, Universit\'e Paris XI,\\ F-91405 Orsay Cedex, France}
\Email{\href{mailto:rivass@th.u-psud.fr}{rivass@th.u-psud.fr}}
\URLaddress{\url{http://www.rivasseau.com/}}

\ArticleDates{Received March 23, 2016, in f\/inal form August 18, 2016; Published online August 22, 2016}

\Abstract{We provide an up-to-date review of the recent constructive program for f\/ield theories of the vector, matrix and tensor type, focusing not on the models themselves but on the mathematical tools used.}


\Classification{81T08}

\renewcommand{\thefootnote}{\arabic{footnote}}
\setcounter{footnote}{0}

\section{Introduction}

This paper is the second one of our twin contributions to the ``Special Issue on Tensor Models, Formalism and Applications'' of the SIGMA journal. In our f\/irst paper~\cite{Rivasseau:2016zco} the physical motivations behind the tensor models and their connection to quantum gravity are discussed. So we advise strongly the reader not already familiar with tensor models and their relationship to quantum gravity to read \cite{Rivasseau:2016zco} f\/irst, as we shall not duplicate the corresponding material here, but provide below only a minimal introduction to tensor models. We shall discuss in detail in this paper only the constructive program for (Euclidean) random tensor models of general rank~$d$, which include as special cases vector (rank 1) and matrix (rank 2) models. For other aspects of the burgeoning study of random tensor models and their associated f\/ield theories we refer to the other contributions to the special SIGMA issue as well.

Tensor models were introduced as promising candidates for an \emph{ab initio} quantization of gravity \cite{ambjorn, ADT2,gross,sasa}. Indeed they are combinatorial objects which do not refer to any background space-time, nor even to any background topology. However such tensors were initially introduced as symmetric in their indices, a feature which for a long time prevented to investigate them with rigorous analytic methods. In particular in contrast with the famous 't Hooft $1/N$ expansion for random matrix models, there was no way to probe the large $N$ limit of such symmetric random tensors at rank $d\ge 3$.
The modern reformulation \cite{Bonzom:2012hw, Gurau:2009tw,Gurau:2011kk,Gurau:2011xp} unlocked the theory by considering \emph{unsymmetrized} random tensors. These objects have a larger, truly tensorial symmetry (typically in the complex case a ${\rm U}(N)^{\otimes d}$ symmetry at rank $d$ instead of the single ${\rm U}(N)$ of symmetric tensors). This larger symmetry allows to probe their large $N$ limit through~$1/N$ expansions of a new type \cite{Bonzom:2012wa,Bonzom:2015axa,Gurau:2010ba,Gurau:2011xq,Gurau:2011aq}.

Random tensor models can be further divided into fully invariant models, in which both propagator and interaction are invariant, and tensor f\/ield theories (hereafter TFT) in which the interaction is invariant but the propagator is not \cite{BenGeloun:2011rc}. This propagator can also incorporate a~further gauge invariance to make contact with group f\/ield theory \cite{Carrozza:2013wda, Carrozza:2012uv}, in which case one calls them tensor group f\/ield theories (TGFT).

We restrict ourselves in this paper to the relatively simple case of models with a quartic interaction, for which the typical constructive task is to prove Borel summability of the free energy or of the connected Schwinger functions on the ``stable side'' of the coupling constant. To be non-trivial and physically interesting, such Borel summability theorems should be either \emph{uniform in $N$}, the size of the tensor, for tensor models, or \emph{uniform in the ultraviolet cutoff}, hence include renormalization, for TFTs or TGFTs.

Our review starts with a section on general mathematical prerequisites: quartic tensor models, Nevanlinna theorem \cite{Sok} and the intermediate f\/ield representation (IFR) plus key combinatorial constructive tools, namely the forest and jungle formulas \cite{Abdesselam:1994ap,BK}. The next section presents three specif\/ic constructive methods adapted to tensor models: the loop vertex expansion (LVE), which is a combination of the intermediate f\/ield representation and of the forest formula f\/irst introduced in~\cite{Rivasseau:2007fr}; the multiscale loop vertex expansion (MLVE), which is a combination of the intermediate f\/ield representation and of a level-2 jungle formula with a bosonic and a fermionic part f\/irst introduced in~\cite{Gurau:2013oqa}; and iterated Cauchy--Schwarz bounds (ICS bounds), introduced in this context in \cite{Magnen:2009at} and~\cite{Delepouve:2014bma}, for which we provide here a slightly improved version.

Then in the last section we review the recent Borel summability results on tensor models in increasing order of dif\/f\/iculty. The quartic melonic tensor model in any~$d$~\cite{Gurau:2013pca} requires a~simple LVE and no ICS bounds. This is also the case for the melonic ${\rm U}(1)-T^4_3$ TGFT \cite{Lahoche:2015yya}. The general quartic tensor model in any $d$ requires a LVE plus ICS bounds \cite{Delepouve:2014bma}. The melonic ${\rm U}(1)-T^4_4$ TGFT requires an MLVE but still no ICS bounds~\cite{Lahoche:2015zya}. Finally the melonic $T^4_3$ TFT requires both MLVE and ICS bounds~\cite{Delepouve:2014hfa}. We discuss also brief\/ly the case of matrix models, for which the LVE has to be completed by an additional topological expansion in order to identify the~$1/N$ expansion of the model in a non-perturbative way up to f\/inite genus $g$ \cite{Gurau:2014lua}.

The world of tensor models and f\/ield theories is extremely vast, as tensor interactions encode inf\/initely many triangulations of any piecewise linear manifold in any dimension and of a large class of pseudo-manifolds. It also includes in particular any matrix model, not only invariant ones, but also non-commutative f\/ield theories such as the Grosse--Wulkenhaar model which we shall discuss brief\/ly too. The constructive program to explore all these models at the rigorous non-perturbative level has clearly just started. However a rich harvest of non-perturbative Borel-summability results has already been obtained for tensor f\/ield theories. Such constructive results essentially have no counterpart yet in the other approaches to quantum gravity. This is very promising for the tensor track towards quantum gravity~\cite{Rivasseau:2013uca}. In our conclusion we indicate what could be the next steps of this constructive program.

\section{Mathematical prerequisites}

\subsection{Quartic tensor models}
For a general introduction to tensor models we refer to \cite{Gurau:2011xp}. We simply recall in this subsection the notations required for the quartic models discussed thereafter, following roughly \cite{Delepouve:2014bma}.

For a complex rectangular random matrix $M$, called a Wishart matrix, the basic invariant homogeneous of order $2p$ is unique. It is $\Tr (M M^\dagger )^p$.
All other polynomial invariants are linear combinations of products of these basic bricks. In particular there is only one\footnote{Of course there is another quartic ``disconnected'' interaction, namely $[\Tr MM^\dagger]^2$. We shall not consider further such ``multi-trace'' models here. See however~\cite{Bonzom:2015axa}.} matrix model with a connected quartic interaction, namely $\Tr M M^\dagger M M^\dagger$. The situation is completely dif\/ferent in the case of tensor models of higher rank $d$, as there are many more available interactions and in particular many dif\/ferent invariant quartic interactions.

We consider a Hermitian inner product space $V$ of dimension $N$ and $e_n$, $ n=1,\dots ,N$ an orthonormal basis in~$V$. A covariant tensor of rank $d$ is a~multilinear form ${\bf T}\colon V^{\otimes d} \to \mathbb{C}$. We denote its components in the canonical dual tensor product basis by
\begin{gather*}
T_{n_1\dots n_d} \equiv {\bf T} (e_{n_1},\dots, e_{n_d}), \qquad {\bf T} = \sum_{n_1,\dots n_d} T_{n_1\dots n_d}  e^{n_1} \otimes \dots \otimes e^{n_d} .
\end{gather*}
A priori $T_{n_1\dots n_d}$ has no symmetry properties, hence its indices have a well def\/ined position. We call the position of an index its \emph{color}, and we denote $\cD$ the set of colors $\{1,\dots, d\}$. Subsets $\cC\subset \cD$ will be called generalized colors; singleton subsets identify with the ordinary colors.

Using the canonical identif\/ication of $V$ with its dual thanks to the Hilbert scalar product, we can also expand the dual tensor as a conjugated multilinear map ${\bf T}^{\vee}$ with matrix elements $ \bar T_{\bar n_1\dots \bar n_d} $ in the canonical dual basis. We write all the indices in subscript, and we denote the contravariant indices with a bar. Indices are always understood to be listed in increasing order of their colors. We denote $ \delta_{n_{\cC} \bar n_{\cC}} = \prod\limits_{c\in \cC} \delta_{n_c \bar n_c} $ and $\Tr_{\cC}$ the partial trace over the indices~$n_c$, $c\in \cC$.

Under unitary base change, covariant tensors transform under the tensor product of $d$ fundamental representations of ${\rm U}(N)$: the group acts independently on each index of the tensor. For $U^{(1)}\cdots U^{(d)}\in {\rm U}(N)$,
\begin{gather*}
 {\bf T} \to \big( U^{(1)} \otimes \cdots \otimes U^{(d)} \big){\bf T} ,
 \qquad {\bf T}^{\vee} \to {\bf T}^{\vee} \big( U^{(1)*} \otimes \cdots \otimes U^{(d)*} \big).
\end{gather*}
In components, it reads
\begin{gather*}
 T_{a_{\cD}}\to \sum_{m_{\cD}} U^{(1)}_{a_1 m_1}\cdots U^{(d)}_{a_d m_d}\ T_{m_{\cD}},
 \qquad \bar{T}_{\bar a_{\cD}}\to \sum_{\bar m_{\cD}} \bar U^{(1)}_{\bar a_1 \bar m_1} \cdots \bar U^{(d)}_{\bar a_d \bar m_d}\ \bar T_{\bar m_{\cD}}.
\end{gather*}

A \emph{tensor invariant} is a polynomial in the components of the tensor and its conjugate which is invariant under this action of the external tensor product of $d$ independent copies of the unitary group ${\rm U}(N)$. Tensor invariants are linear combinations of products of connected invariants (also called bubbles or generalized traces). Such connected invariants are built by contracting indices of the same color of a product of tensor entries into a connected graph. Hence connected invariants are in one to one correspondence with $d$-regular edge-colored bipartite connected graphs~\cite{Gurau:2011xp} and can be enumerated quite precisely~\cite{Geloun:2013kta}.

The unique quadratic trace invariant is the (scalar) Hermitian pairing of ${\bf T}^{\vee}$ and ${\bf T}$ which reads
\begin{gather*}
 {\bf T}^{\vee} \cdot_{\cD } {\bf T} = \sum_{n_{\cD} \bar n_{\cD}}
 \bar T_{ \bar n_1\dots \bar n_d} \delta_{ \bar n_{\cD} n_{\cD}} T_{ n_1\dots n_d} ,
\end{gather*}

In arbitrary rank $d$, the most general quartic connected trace invariants are associated to non trivial generalized colors, i.e., to subsets $\cC \subset \cD$, $\cC \ne \varnothing$, $\cC \ne \cD $. The connected quartic inva\-riant~$V_{\mathcal{C}}$ is
\begin{gather*}
 V_{\cC}({\bf T}^{\vee},{\bf T} ) = \Tr_{\cC} \big[ \big[ {\bf T}^{\vee} \cdot_{\cD\setminus \cC } {\bf T} \big] \cdot_{\cC}
 \big[{\bf T}^{\vee} \cdot_{\cD\setminus \cC } {\bf T} \big] \big]  ,
\end{gather*}
where we denoted $\cdot_{\cC}$ the product of operators from $V^{\otimes \cC}$ to $V^{\otimes \cC}$. In components this invariant writes:
\begin{gather*}
 \sum_{n, \bar n, m, \bar m} \big( \bar{T}_{\bar{n}} \delta_{\bar n_{\cD\setminus \cC} n_{\cD\setminus \cC} }  T_{n} \big)
 \delta_{n_{\cC}\bar m_{\cC}} \delta_{ \bar n_{\cC} m_{\cC}}
 \big( \bar{T}_{\bar{m}} \delta_{\bar m_{\cD\setminus \cC} m_{\cD\setminus \cC} }   T_{m} \big)   .
\end{gather*}
Remark that the invariant for $\cC$ is the same as for its complement $\cD\setminus \cC$. Hence at rank $D$
there are exactly $2^{D-1} - 1$ quartic trace invariants. The 7 invariants at rank $d=4$ are shown in Fig.~\ref{order4}.

\begin{figure}[t]\centering
\includegraphics[scale=.4]{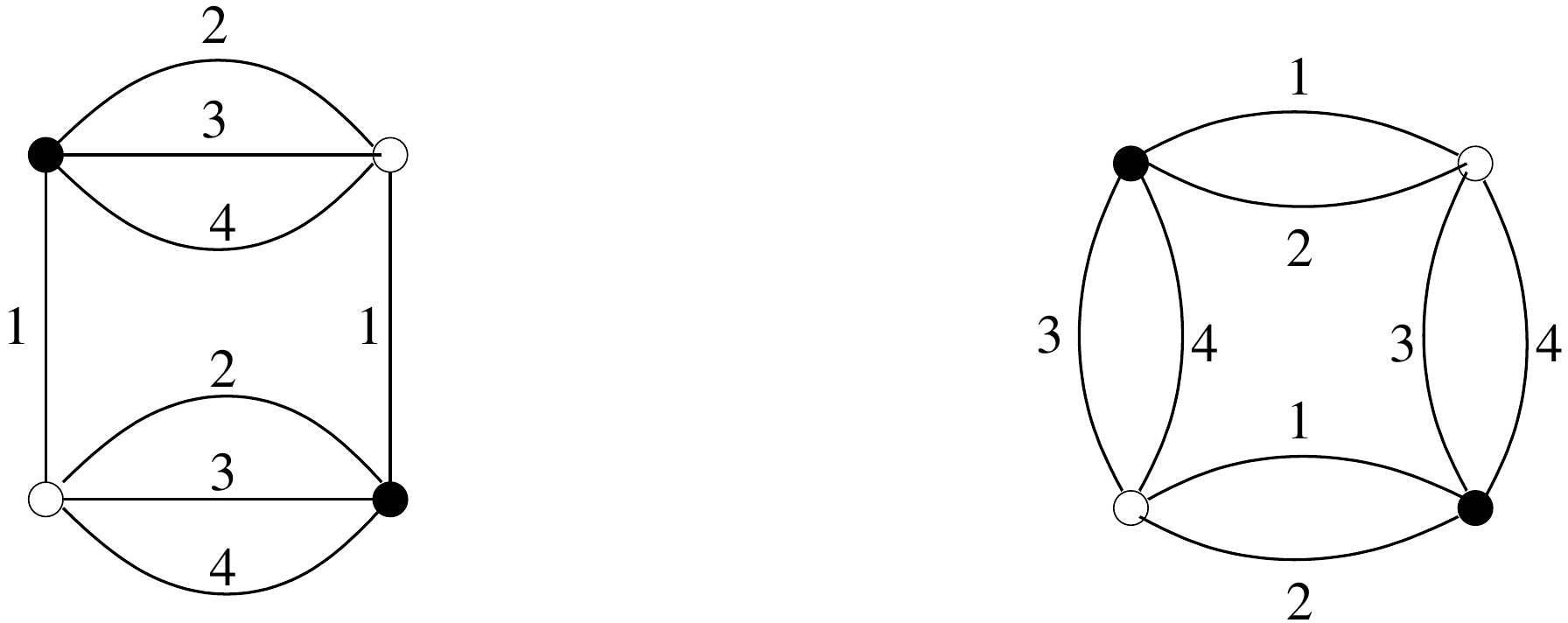}
\caption{The quartic invariants of rank~4 \cite{Delepouve:2014bma}: there are four types of ``melonic'' ones (one of them shown left, the others being obtained by permuting the colors) and three types of ``necklaces'' (one of them shown right, the others being obtained by color permutations).} \label{order4}
\end{figure}

A generic quartic tensor model (with standard scaling and interaction set $\mathcal{Q}$) is then def\/ined by the moments of the (invariant) perturbed Gaussian measure:
\begin{gather}\label{eq:model}
d\mu = \left( \prod_n N^{D-1} \frac{ d \bar{T}_{{n}} dT_n }{2 \imath \pi} \right)
 e^{ -N^{D-1} \big( {\bf T}^{\vee} \cdot_{\cD } {\bf T}
 +\lambda\sum\limits_{\mathcal{C} \in \mathcal{Q}} V_{\mathcal{C}}({\bf T}^{\vee},{ \bf T} ) \big) }  ,
\end{gather}
where $\mathcal{Q}$ is some subset of $\cC$s. The global scaling factor $N^{D-1}$ for the action is the unique one such that the model admits a non-trivial $1/N$ expansion~\cite{Bonzom:2015axa, Gurau:2011xp}. Of course more general models can be considered by giving an independent coupling to each interaction, or by considering dimensions $N_c$ depending on the color, or by adding the disconnected interaction $[ {\bf T}^{\vee} \cdot_{\cD } {\bf T}]^2$, but we shall not consider them here.

The ``melonic'' quartic models \cite{Delepouve:2014bma} are obtained by restricting the interactions to singletons, so that $\mathcal{Q}=\{\mathcal{C},\, |\mathcal{C}|=1 \}$ in \eqref{eq:model}. In other words interaction is restricted to the $d$ quartic melonic invariants, one for each color.

Quartic tensor f\/ield theories have the same interactions but a non-symmetric propagator. The simplest such models are built around the Hilbert space $V := L^2 ({\mathbb Z})$ and as in usual f\/ield theory, the propagator is the inverse of a~Laplacian operator plus a mass term
\begin{gather}\label{tftpropa} C (n_\cD, \bar n_\cD) = \frac{\delta_{n_\cD, \bar n_\cD}}{\sum\limits_{j=1}^d n_j^2 + m^2} .
\end{gather}
The melonic quartic TGFT at rank $d$ with propagator \eqref{tftpropa} and a unique coupling constant $\lambda$ identical for each coupling is nicknamed~$T^4_d$. We often further simplify the model by putting $m^2 = 1$.

Finally tensor group f\/ield theories are built around a more general space $V := L^2 (G)$ where $G$ is a Lie group \cite{Carrozza:2013wda, Carrozza:2012uv} or homogeneous space~\cite{Lahoche:2015tqa}, and their propagator incorporates a~Boulatov projector
\begin{gather}\label{tgftpropa}
 C (n_\cD, \bar n_\cD) = \frac{\delta_{n_\cD, \bar n_\cD}\delta \Big(\sum\limits_{j=1}^d n_j\Big) }{\sum\limits_{j=1}^d n^2_j + m^2},
\end{gather}
where the indices $n$ now run over a basis of the dual Fourier space of $L^2 (G)$, and $ n^2_j $ is a sloppy notation for the Laplace--Beltrami operator on~$G$ (see \cite{Carrozza:2013wda} for the case $G={\rm SU}(2)$). The melonic quartic TGFT at rank $d$ with propagator~\eqref{tgftpropa} and Lie group~$G$ will be nicknamed~$G-T^4_d$. For instance in the case $G = {\rm U}(1)$ the Hilbert space~$V$ is $ L^2 ({\mathbb Z})$ as in the previous case. We often consider discrete tensor indices as momenta, and the equivalent representation with continuous indices in $G^d$ as the corresponding ``direct space'' representation.

\subsection{Borel summability}\label{subsec:Borel}

The perturbative expansion in quantum f\/ield theory is obtained by performing a Taylor expansion of the interaction, and then illegally commuting the sum with the Gaussian integral to obtain a series indexed by Feynman graphs. Let us return in this section to ordinary quantum f\/ield theory, since we want to discuss the relationship between functional integrals and their perturbative expansion in the traditional context before turning to tensor models. The ordinary~$\phi^4_d$ theory in dimension~$d$ has propagator kernel
\begin{gather*}
C(x,y) = \int d^d p e^{i p (x-y) }\big(p^2 + m^2\big)^{-1}
\end{gather*}
and its partition function is
\begin{gather} \label{parti}
Z( \lambda) = \int d\mu_C  \, e^{ - \lambda \int d^dx \phi^4(x) } = \sum_{n=0}^{\infty} \frac{(-\lambda)^n}{n!}  \sum_{G} A_G   ,
\end{gather}
where $G$ runs over labeled Feynman graphs\footnote{Labeled Feynman graphs can be considered as def\/ined by Wick contractions of labeled f\/ields belonging to labeled vertices. Hence we shall avoid in this brief review the subtle issue of automorphim groups of graphs and of their corresponding symmetry factors.} with $n$ vertices. The Feynman amplitude of~$G$ is
\begin{gather*}
A_G  =  \int \prod_{v \in V(G)} d^d x_v   \prod_{\ell \in E(G)} C(x_{v(\ell)}, x_{v'(\ell)}),
\end{gather*}
where the index $v$ runs over the set of $n$ vertices of $G$, the index $\ell $ runs over edges of~$G$ and $(v(\ell), v'( \ell))$
is the pair of end vertices of edge~$\ell$. Of course such amplitudes may diverge if $d \ge 2$ (see, e.g.,~\cite{Gurau:2014vwa} for a recent review).

Interesting quantum f\/ield theories have many degrees of freedom, in which case the important physical quantities are intensive rather than extensive quantities, such as the specif\/ic free energy and the connected functions or cumulants~$G^c_{2p}$. Therefore the important quantum f\/ield theory quantity is not~$Z$ but its logarithm. In particular the connected Green functions are obtained from the generating functional
\begin{gather*}
W(J) = \log [Z( \lambda ,J) ]
\end{gather*} through
\begin{gather*}
G^c_{2p}(x_1,\dots, x_{2p}) =\frac{\partial^{2p} W(J)}{\partial J_1 (x_1) \cdots \partial J_{2p} (x_{2p}}\Big|_{J=0},
\end{gather*}
where $Z( \lambda, J)$ is def\/ined by adding a source f\/ield $J$ to the partition function
\begin{gather} \label{partiJ}
Z( \lambda, J) = \int d\mu_C \, e^{\int d^dx \phi(x)J(x) - \lambda \int d^dx \phi^4(x) }.
\end{gather}
Similarly the free energy $p$ also is deduced from the logarithm of $Z$, with interaction restricted to a f\/inite volume $\Lambda$. One should divide by the volume cutof\/f $\Lambda$, and f\/inally remove this cutof\/f, hence $p=\lim\limits_{\vert \Lambda \vert \to \infty} \frac{1}{\vert \Lambda \vert} \log Z_\Lambda(\lambda) $.

The most important property of the perturbative expansion is that it allows for a quick computation of such physical connected quantities. They are indeed given by the same sums as~\eqref{parti}, with the same amplitudes, but \emph{restricted to connected graphs}. For instance the free energy is given in perturbation theory by a sum identical to~\eqref{parti} but restricted to \emph{connected} rooted Feynman vacuum graphs, in which for each graph a particular root vertex has been f\/ixed to the origin. This rule accounts for the need  to break the translation invariance of the theory and divide by the volume. Most of the time in this review we shall consider for simplicity only this free energy.

Perturbative series typically diverge. In particular for $d=0$, in which $A_G$ is 1 for any~$G$, there are $(4n-1)!!$ labeled Feynman graphs with $n$ vertices, hence the perturbative series for~$Z(\lambda)$ is
\begin{gather*}
 \sum_{n=0}^{\infty} \frac{(-\lambda)^n}{n!}(4n-1)!! = \sum_{n=0}^{\infty} \frac{(-\lambda)^n}{n!} \frac{(4n)!}{2^{2n}(2n)!},
\end{gather*}
which has zero radius of convergence. Up to dimension $d\le 3$ the \emph{renormalized}~$\phi^4$ series can be shown rigorously to also have zero radius of convergence~\cite{CR,J}, and it is believed to be also true for $d=4$, although, to our knowledge, there is no proof of this, and the best result in this direction remains the quite old reference~\cite{DFR}.

In the quantum f\/ield theory literature we often read that ``since the perturbation series diverges it can be asymptotic at best", and that ``perturbation theory cannot capture non-perturbative ef\/fects, such as instantons''. Such statements are both confusing and mathematically wrong.

First, asymptoticity of a series is a very weak notion which has nothing to do with its \emph{convergence}. \emph{Every} power series, no matter how convergent or horribly divergent, is asymptotic to \emph{infinitely many} smooth functions of a real-variable. Hence an asymptotic series never encodes the full information about such a function without additional conditions. The typical example is the function $f(\lambda) =0$ for $\lambda \le 0$ and $f(\lambda) = e^{-\frac{1}{\lambda}}$ for $\lambda >0$ whose asymptotic series at $\lambda=0$ is~$0$, the most convergent of all series. It certainly does not encode the information to reconstruct~$f$.

Second, as we shall clarify below, in the case of a Borel summable series, all the information contained in the functional integral (including any non-perturbative ef\/fect) is in fact entirely contained in the coef\/f\/icients of the perturbative series. It is simply dif\/f\/icult to extract in practice.

To correctly relate a function to its power series, the key is to impose an analyticity condition, hence to consider functions of a complex variable. Analyticity is a very strong and rigid condition. For instance the function 0 is the only one asymptotic to the~0 series \emph{in the class of functions analytic in an open neighborhood of the expansion point}.

Fortunately analyticity is a common feature in quantum f\/ield theory. The partition func\-tion~\eqref{parti} is usually analytic in a non-empty domain of the coupling constant. It is certainly analytic for $\Re \lambda>0$ in dimension $d=0$, or if we include volume and ultra-violet cutof\/fs, since in this case the functional integral is uniformly convergent. For $\vert \lambda \vert $ suf\/f\/iciently small this analyticity extends to $\log Z$, since~$Z$ is close to~1 and its logarithm well-def\/ined; note however that the domain of analyticity for $\log Z$ typically shrinks as cutof\/fs are removed.

In the super-renormalizable domain for $\phi^4$, namely in dimension 2 or 3, the point $\lambda=0$ belongs to the boundary of the analyticity domain for the renormalized free energy~\cite{CR, J}. Nevertheless the functional integral~\eqref{parti} can be fully recovered from the perturbative series through {\it Borel summability} \cite{EMS,MS}. The most useful Borel summability theorem for quantum f\/ield theory is due to Nevanlinna~(1919) and was rediscovered by A.~Sokal:

\begin{thm}[Nevanlinna--Sokal \cite{Sok}]
 A function $f(\lambda)$ with $\lambda\in \mathbb{C}$ is said to be Borel summable in $\lambda$ if
 \begin{itemize}\itemsep=0pt
 \item $f(\lambda)$ is analytic in a disk ${\cal D}_R := \big\{ \lambda \, | \, \Re\big(\lambda^{-1}\big)>R^{-1}\big\}$ with $R\in \mathbb{R}_+$.
 \item There exists some constants $K$ and $\sigma$ such that $f(\lambda)$ expands at the origin with uniform Taylor remainder bounds in the disk ${\cal D}_R$
 \begin{gather*}
 f(\lambda) = \sum_{ k =0}^{n-1} f_{k} \lambda^k + R_{n}(\lambda)  , \qquad |R_{n}(\lambda)| \le K \sigma^n n! |\lambda|^n .
 \end{gather*}
 \end{itemize}

If $f(\lambda)$ is Borel summable in $\lambda$ then
\begin{gather*}
 B(t) =\sum_{k=0}^{\infty} \frac{1}{k!} f_{k} t^k
\end{gather*}
is an analytic function for $|t|<\sigma^{-1}$ which admits an analytic continuation in the strip
 $\{ z\, | \, | \Im z | < \sigma^{-1} \} $ such that $|B(t)| \le B e^{t/R}$ for some constant $B$ and $f(\lambda)$ is
given by the absolutely convergent integral
 \begin{gather*}
 f(\lambda) = \frac{1}{\lambda} \int_0^{\infty} dt \, B(t) e^{-\frac{t}{\lambda}} .
 \end{gather*}
\end{thm}

That is the Taylor expansion of $f(\lambda)$ at the origin is Borel summable, and $f(\lambda)$ is its Borel sum. The set $\big\{\lambda\, |\, \Re \big(\lambda^{-1}\big) >R^{-1},\, R\in \mathbb{R}_+ \big\} $ is a disk in the complex plane with center at $\frac{R}{2}$ and of radius $\frac{R}{2}$ (see Fig.~\ref{fig:borel}) as, denoting $\lambda = \frac{R}{2}+ae^{\imath \gamma}$,
\begin{gather*}
\Re \big(\lambda^{-1}\big)>R^{-1} \Leftrightarrow \frac{R^2}{4} > a .
\end{gather*}

\begin{figure}[t]\centering
 \includegraphics[width=3cm]{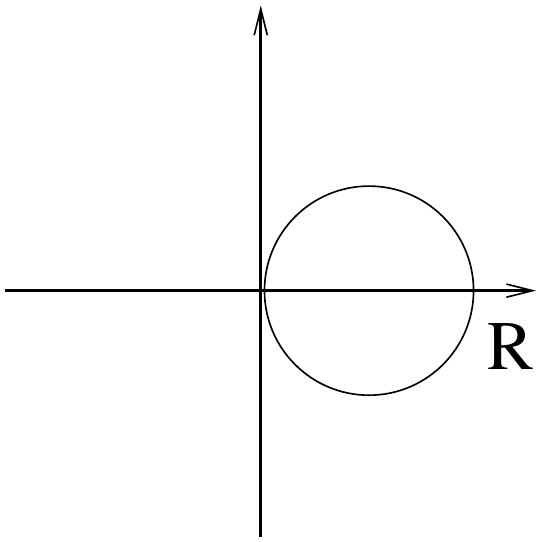}
 \caption{A disk for Nevanlinna's theorem.} \label{fig:borel}
\end{figure}

The important remark is that Borel summability provides a {\it uniqueness} criterion: if a power series with coef\/f\/icients~$a_n$ is the Taylor expansion of a Borel summable function $f(\lambda)$ at $\lambda=0$, then $f$ is \emph{uniquely defined} by the list of $a_n$. This is the case for the partition function of the $\phi^4_d$ model with f\/inite cutof\/fs, and also for the connected functions without cutof\/fs after renormalization in the super-renormalizable
domain $d \le 3$ \cite{EMS,MS}. It is also the case for the ``infrared" limit of $\phi^4_4$ with f\/ixed ultraviolet cutof\/f and at the critical point,
which is a just renormalizable marginal theory asymptotically free in the infrared direction \cite{FMRS1987}.

When Borel summability holds, the perturbative expansion in fact contains (although in a~quite hidden manner) all information about the functional integral, \emph{including instanton effects to all orders}, since a good def\/inition of these instantons is that they occur as discontinuities between dif\/ferent branches of analytic continuation of the functional integral~\cite{BergereDavid}. Of course such ef\/fects are very dif\/f\/icult to extract in practice from the perturbative coef\/f\/icients, because analytic continuation of a function def\/ined by its Taylor germ at a point is typically very dif\/f\/icult.

In summary Borel summability is a key step for the construction of a quantum f\/ield theory model, since when it holds, it def\/ines a unique analytic germ synthesizing Feynman's two great ideas: the functional integral and the Feynman graphs. It does not solve all physical issues, such as computing the long range behavior of the model and its bound states, but it renders such questions at least mathematically \emph{well-defined}.

\subsection{Intermediate field representation (IFR)}

Any quartic (i.e., four-body) interaction can be decomposed as the gluing of two more elementary three body-interactions, by splitting the vertex through a so-called intermediate f\/ield. This method (usually called Hubbard--Stratonovic transformation in condensed matter, and Matthews--Salam representation in high energy physics) has been extremely fruitful in physics. In the case of the ordinary~$\phi^4$ theory def\/ined by~\eqref{partiJ}, with a slightly more adapted normalization for the coupling $\lambda$, it consists in writing
\begin{gather*}
e^{ -\frac{\lambda}{2} \int d^dx \phi^4(x) }
 = \int d\nu (\sigma) e^{ i \sqrt \lambda \int d^dx \phi^2(x) \sigma (x) },
\end{gather*}
where $d\nu$ is the normalized Gaussian measure with ``ultralocal'' covariance kernel $\delta(x-y)$. The functional integral over $\phi$ becomes Gaussian, hence can be explicitly performed, leading to
\begin{gather*}
Z(\lambda, J) = \int d\nu (\sigma ) \, e^{\langle J,C^{1/2} R(\sigma) C^{1/2} J\rangle } e^{ - \frac{1}{2}\Tr \log (1 - i \sqrt \lambda C^{1/2} \sigma C^{1/2}) },
\end{gather*}
where we used operator notation. $\sigma$ is a local multiplication operator, diagonal in direct space, $C^{1/2}$ is the square root in operator sense of the positive operator~$C$, hence is not diagonal in direct space\footnote{Cyclicity of the trace means we can also freely replace the operator $C^{1/2} \sigma C^{1/2}$, by $C \sigma$ or $\sigma C$, leading to slightly shorter formulas, but this would be unwise as Hermiticity would no longer be visible.}.
1 means the identity operator (with kernel $\delta(x-y)$ in direct space), and $\Tr K = \int d^dx K(x,x)$ simply means integration of the diagonal part of a~kernel~$K$. Finally the resolvent $R (\sigma)$ is given by
\begin{gather*} R (\sigma) = \big[ 1 - i \sqrt \lambda C^{1/2} \sigma C^{1/2}\big]^{-1},
\end{gather*}
and for any $\sigma$ and any $\lambda$ real positive it is a well def\/ined operator norm-bounded by~1. More precisely a moment of investigation reveals that
the useful natural domain of def\/inition and analyticity in the~$\lambda$ complex plane for such a resolvent is a cardioid, pictured in Fig.~\ref{cardioidfig}.

\begin{figure}[t]\centering
\includegraphics[width=5cm]{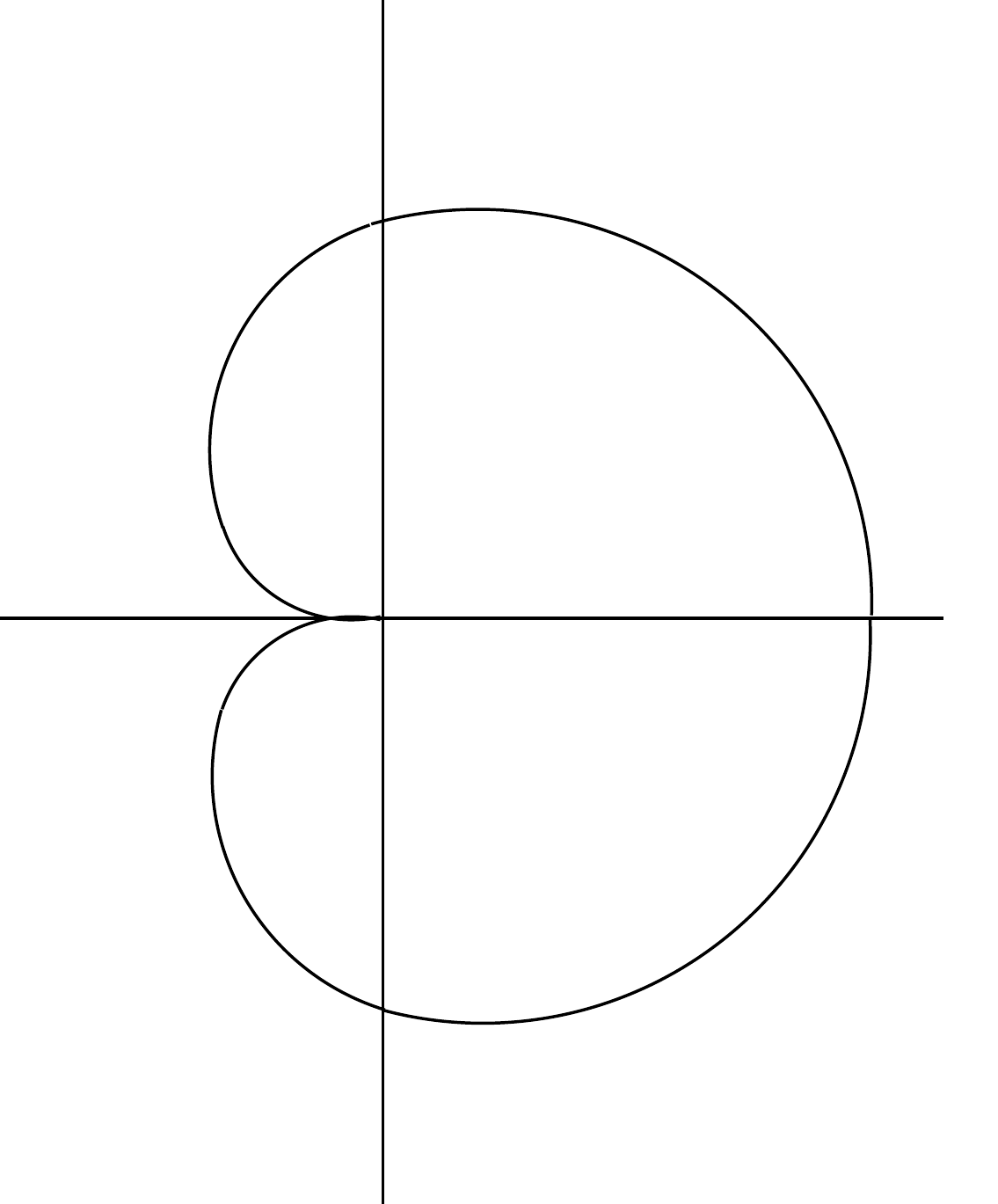}
\caption{The cardioid domain.}\label{cardioidfig}
\end{figure}

This cardioid is def\/ined, in polar coordinates $\lambda = \rho e^{i \phi}$, $\rho >0$, $\phi \in ]{-}\pi , \pi[ $ by the inequality $\rho < [ \cos (\phi/2 )]^{2} $. Indeed under this condition we have that $R (\sigma)$ is a well def\/ined operator norm-bounded by $[ \cos (\phi/2 )]^{-1}$. In expansion steps of the resolvents such as
\begin{gather*} R (\sigma) = 1 + i \sqrt \lambda C^{1/2} \sigma C^{1/2} R (\sigma)
\end{gather*}
the bad factor $[ \cos (\phi/2 )]^{-1}$ for $\Vert R\Vert$ will be always compensated by a good factor $\sqrt \rho < \cos (\phi/2 )$ coming from the $\sqrt \lambda$ numerator. Remark that any cardioid domain \emph{contains} a Nevanlinna disk of the type shown in Fig.~\ref{fig:borel}.

Just like the ordinary representation, the IFR has an associated perturbation theory, obtained by expanding the exponential of the (non-polynomial) interaction $V=\frac{1}{2}\Tr \log \big( 1 - i \sqrt \lambda C^{1/2} \sigma C^{1/2} \big)$ as
\begin{gather*} e^{-V} =\sum_{n=0}^\infty \frac{(-V)^n}{n!},
\end{gather*}
then commuting (again illegally!) integration with respect to $d\nu$ and Gaussian integration over~$\sigma$, leading to
\begin{gather*} Z(\lambda) = 1 + \sum_{n=1}^\infty \frac{(-1)^n}{n!} \Big[ e^{ \frac{1}{2} \int d^dx \frac{\delta^2}{\delta \sigma (x)^2}} V(\sigma)^n \Big]_{\sigma =0} .
\end{gather*}

An intermediated f\/ield propagator is represented in Fig.~\ref{figinterfey0} below by a dashed line, which corresponds to the (ultra-local) covariance of a $\sigma$ intermediate f\/ield. The interaction $V=\frac{1}{2}\Tr \log \big(1 - i \sqrt \lambda C^{1/2} \sigma C^{1/2}\big) $ is non polynomial, but since $\Tr \log 1 = 0$, it appears in the perturbative expansion only through its functional derivatives which we call \emph{loop vertices}. In contrast with the original~$\phi^4$ vertices, such loop vertices have arbitrary coordination $q\ge 1$, and bear the initial propagators of the theory on their~$q$ \emph{corners} or arcs.

Several interesting remarks are in order. The IFR exchanges the traditional role of propagators and vertices. The IFR loop vertices are non-local, as they contain the propagators of the ordinary theory. The IFR propagators in contrast are local as they correspond to the vertices of the ordinary theory. An even deeper remark is that the cyclic character of the trace allows for an embedding of the graphs in the plane, so that the Feynman graphs of the IFR should really be considered as \emph{combinatorial maps}.

\begin{figure}[t]\centering
\includegraphics[scale=0.4]{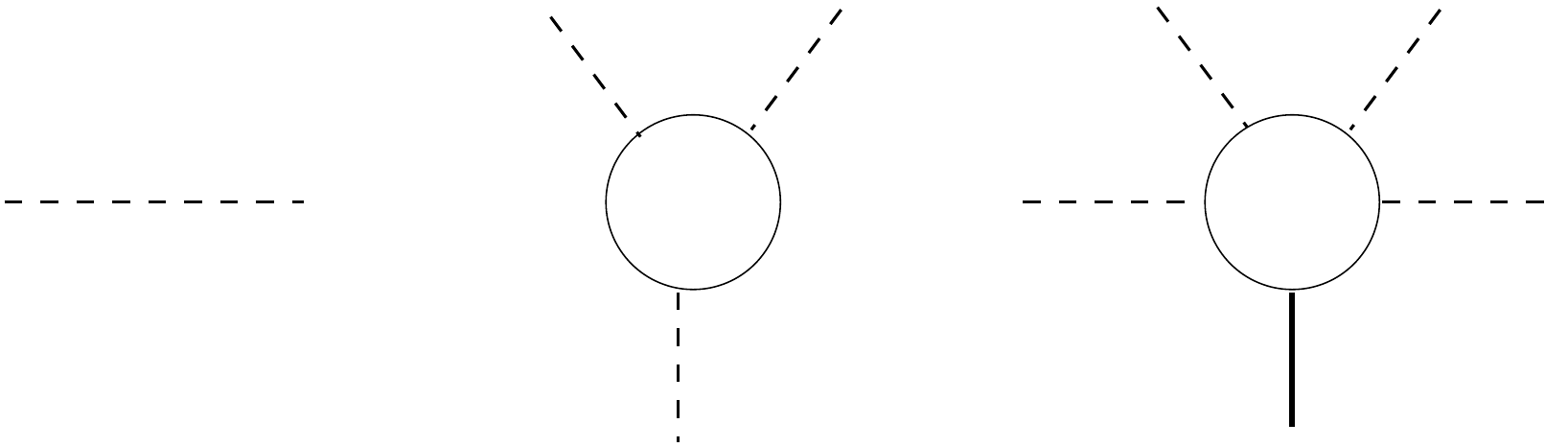}
\caption{An intermediate f\/ield propagator (dashed line), a loop vertex (circle, here with~3 corners or arcs and~3 intermediate f\/ield half-propagators) and a ciliated loop vertexwith 4 intermediate f\/ield half-propagators and f\/ive arcs/corners.}\label{figinterfey0}
\end{figure}

Finally it is easy to add sources to the picture. Any connected function of order $2p$ is obtained by the same expansion in which we add to the ordinary loop vertices $p$ particular \emph{ciliated loop vertices}~\cite{Gurau:2013pca}. The cilium of the $j$-th ciliated vertex means a~cut with two dif\/ferent external variables $x_j$, $x_{j+1}$ for $j=1, \dots , p$ on both sides of the cilium. Apart from the cut, ciliated vertices are similar to loop vertices and bear propagators on each of their corners.

\begin{figure}[t]\centering
\includegraphics[scale=0.4]{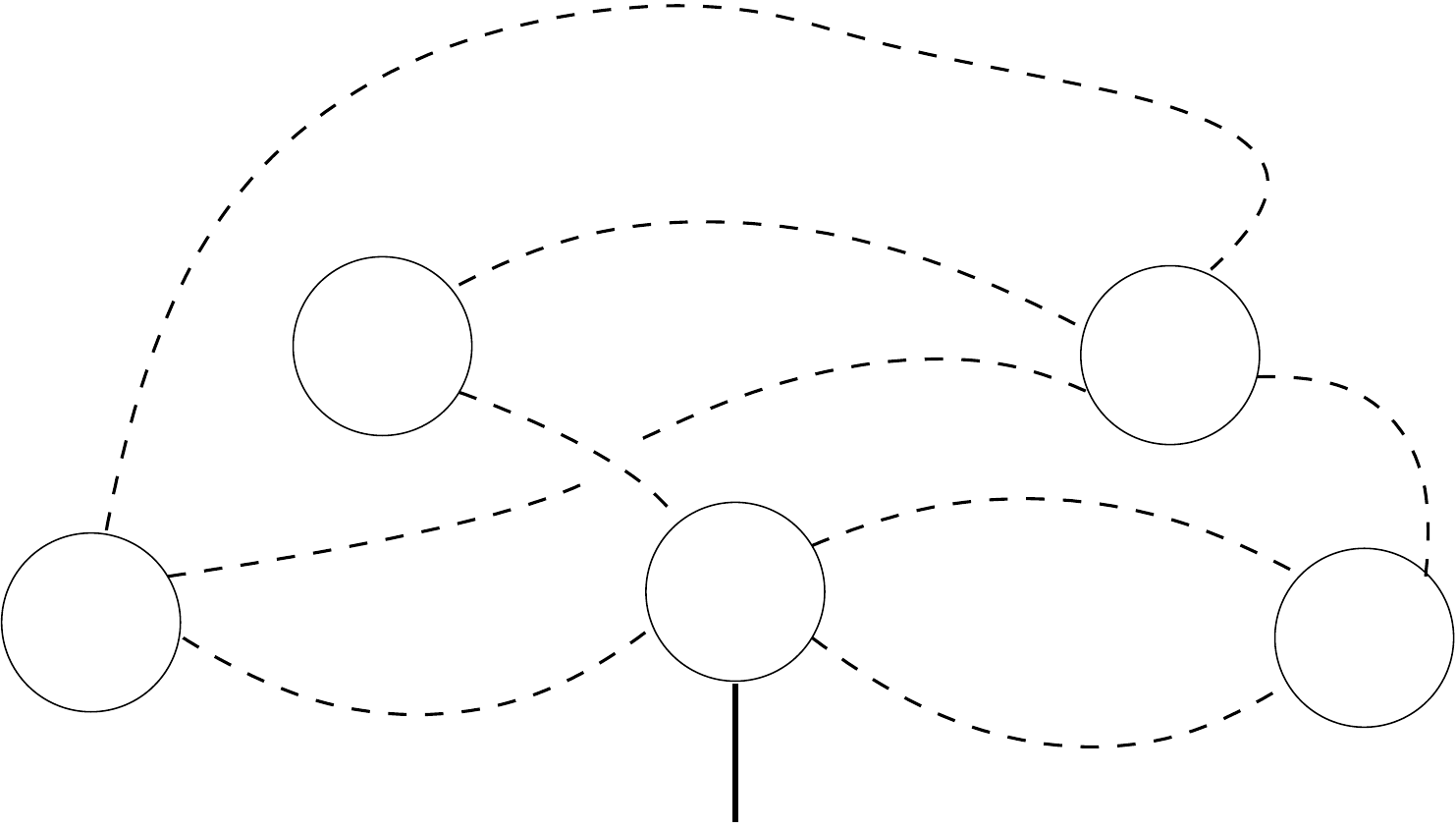}
\caption{A 2-point graph in the intermediate f\/ield theory, here of order 8, has a single \emph{ciliated} loop vertex.}\label{figinterfey2}
\end{figure}

The $2p$-point functions of the theory become sums over Feynman graphs of the intermediate representation with exactly~$p$ ciliated vertices, as shown in Fig.~\ref{figinterfey2}. The order of perturbation is simply the number of dashed lines.

\subsection{The forest formula}\label{forformul}

A forest formula expands a quantity def\/ined on $n$ points in terms of forests built on these points. Forest formulas, provided they have a positivity property are the key combinatorial component of constructive f\/ield theory and were developed systematically in particular by Brydges, Battle and Federbush. The most beautiful such formula is symmetric under action of the permutation group on the~$n$ points. It was discovered in~\cite{BK} and developed with alternative proofs in~\cite{Abdesselam:1994ap}\footnote{Non-symmetric versions appeared earlier in the constructive literature, but won't be treated here (see~\cite{Rivasseau:2013tpa} for a recent reference).}.

Consider $n$ points which we identify with the set $V_n$ of vertices of the complete graph $K_n$; the set of pairs of such points has $n(n-1)/2$ elements $\ell = (i,j)$ for $1\le i < j \le n$ and can be identif\/ied with the set $E_n$ of edges $\ell$ of $K_n$. The forest formula is often presented as a Taylor expansion for functions~$f$ of $n(n-1)/2$ variables $x_\ell$, $\ell \in E_n$ which are smooth, e.g., on an open neighborhood of $[0,1]^{n(n-1)/2}$. Here we expose a variant borrowed from \cite{Gurau:2014vwa} close to constructive applications.

Consider the vector space $S_n$ of symmetric $n$ by $n$ matrices $X= \{X_{ij}\}$, $i,j = 1 ,\dots, n$. It has dimension $n(n+1)/2$. The set ${\rm PS}_n$ of positive symmetric matrices whose diagonal coef\/f\/icients are all equal to 1 and of\/f-diagonal elements are between 0 and 1 is compact and convex. Symmetric matrices with diagonal elements equal to one and of\/f-diagonal elements in $[0,1]^{n(n-1)/2}$ do not all belong to ${\rm PS}_n$, for instance the matrix $\begin{pmatrix} 1 &1&0 \cr 1 &1&1 \cr 0 &1&1 \end{pmatrix}$ is not positive. Any matrix $X\in {\rm PS}_n$ can be parametrized by $n(n-1)/2$ elements~$X_\ell$, where $\ell$ runs over the edges of the complete graph~$K_n$.

${\rm PS}_n$ contains as particularly interesting elements the block matrices $X^\Pi$ for any partition $\Pi$ of $V_n$. The block matrix $X^\Pi$ has entries $X_{ij}^\Pi =1 $ if $i$ and $j$ belong to the same block of the partition~$\Pi$, and~0 otherwise. Two extremal cases are the identity matrix $\cId$, corresponding to~$X^{\rm sing}$, that is to the maximal partition made of all singletons, and the matrix $\bbone$ with all entries equal to one, corresponding to~$X^{V_n}$, that is to the minimal partition made of a single block.

Let us consider a function $f$ def\/ined and smooth in the interior of ${\rm PS}_n$ with continuous extensions (together with all their derivatives) to ${\rm PS}_n$ itself. The forest formula can be expressed as a multi-variate Taylor formula with integral remainder which expands such a~function between the minimal and maximal block-partition matrices~$\bbone$ and~$\cId$. The important point is that the Taylor remainder integrands stay on the ${\rm PS}_n$ convex set.
The precise statement is

\begin{thm}[the forest formula]\label{theoforest}
\begin{gather*} f( {\bbone }) = \sum_{\cF} \int dw_\cF \, \partial_\cF f   \big[ X^\cF (w_\cF) \big], \label{bkar}
\end{gather*}
where
\begin{itemize}\itemsep=0pt
\item the sum over $\cF$ is over forests over $n$ labeled vertices $i = 1, \dots , n$, including the empty forest with no edge. Such forests
are exactly the acyclic edge-subgraphs of the complete graph $K_n$,

\item $\int dw_\cF$ means integration from~$0$ to $1$ over one parameter for each forest edge: $\int dw_\cF \equiv \prod\limits_{\ell\in \cF} \int_0^1 dw_\ell $. There is no integration for the empty forest since by convention an empty product is~$1$. A generic integration point $w_\cF$ is therefore made of $\vert \cF \vert$ parameters $w_\ell \in [0,1]$, one for each $\ell \in \cF$,

\item $ \partial_\cF = \prod\limits_{\ell\in \cF} \partial_\ell $ means a product of first-order partial derivatives with respect to the variables~$X_{\ell}$ corresponding to the edges of~$\cF$. Again there is no such derivatives for the empty forest since by convention an empty product is~$1$,

\item $X^\cF (w_\cF)$ is defined by $X^\cF_{ii} (w_\cF )= 1$ $\forall\, i$, and for $i \ne j$, $X^\cF_{ij} (w_\cF)$ is the infimum of the $w_\ell$ parameters for $\ell$ in the unique path $P^\cF_{i \to j}$ from $i$ to $j$ in $\cF$, when such a path exists. If no such path exists, which means that $i$ and $j$ belong to different connected components with respect to the forest $\cF$, then by definition $X^\cF_{ij} (w_\cF) = 0$,

\item the symmetric $n$ by $n$ matrix $X^\cF (w_\cF)$ defined in this way is positive, hence belongs to ${\rm PS}_n$, for any value of $w_\cF$.
\end{itemize}
\end{thm}

Since $X^\varnothing = \cId$, the empty forest term in \eqref{bkar} is $f(\cId)$, hence \eqref{bkar} indeed
interpolates $f$ between $\bbone$ and $\cId$, staying on ${\rm PS}_n$ as announced.

\begin{proof} We would like to add a new proof to the seven proofs of \cite{Abd1001}, hence a new item
among the 1001 proofs which should exist according to that reference. We thank D.~Brydges for suggesting this proof.

We recall f\/irst Kruskal's greedy algorithm \cite{kruskal}. Consider a f\/ixed connected graph $G$, possibly with self-loops and multiple edges. For any \emph{Hepp sector} $\sigma$, hence any complete ordering of the edges of~$G$, this algorithm def\/ines a unique particular ``leading tree'' $T( \sigma)$,
which minimizes $\sum\limits_{\ell \in T} \sigma (\ell)$ over all trees of~$G$, where $\sigma(\ell)$ is the order of~$\ell$ in~$\sigma$. We call $T( \sigma)$, the \emph{leading tree} for $\sigma$. The algorithm simply picks the f\/irst edge $\ell_1$ in $\sigma$ (i.e., whose order $\sigma(\ell)$ is minimum) which is not a~self-loop. Then it picks the next edge $\ell_2$ in $\sigma$ that does not add a cycle to the (usually disconnected) graph with vertex set $V$ and edge set $\ell_1$ such that the order $\sigma(\ell_2)$ is minimal among edges with this property, and so on. Another way to explain the algorithm is through a deletion-contraction recursion: following the ordering of the sector~$\sigma$, every edge is either deleted if it is a self-loop or contracted if it is not. The set of contracted edges is exactly the leading tree~$T(\sigma)$.

Returning to the proof of the forest formula, we f\/irst bluntly develop the function $f$ at f\/irst order with integral remainder over \emph{all} its parameters, hence over all edges of the complete graph~$K_n$. The result is simply in notations compatible with Theorem~\ref{theoforest}
\begin{gather*} f( {\bbone }) = \sum_{S \subset K_n} \int dw_S \,\partial_S f \big( Y^S (w_S)\big ),
\end{gather*}
where $Y^S_{\ell} (w_S)$ is simply $w_\ell$ if $\ell \in S$ and 0 if $\ell \not\in S$.

Now for a given $S$ we decompose the $w_S$ integrals according to all Hepp sectors of $S$
\begin{gather*} f( {\bbone }) = \sum_{S \subset K_n} \sum_{\sigma} \int_{\sigma} dw_S \, \partial_S f \big( Y^S (w_S) \big),
\end{gather*}
and regroup all terms according to the leading forest $\cF(S, \sigma)$ which is made of the leading Kruskal tree in each connected component~$S_j$ of~$S$ for the order on $S_j$ induced by~$\sigma$. In this way we get
\begin{gather}\label{formul1} f( {\bbone }) = \sum_\cF \sum_{S \subset K_n, \sigma \atop \cF(S, \sigma) =\cF}
\int_{\sigma} dw_S\, \partial_S f \big( Y^S (w_S) \big)  .
\end{gather}

Now to analyze the result for a f\/ixed $\cF$ we remark that any edge $\ell \in K_n-\cF$ which does not add a cycle to $\cF$ cannot belong to any~$S$ in \eqref{formul1}, hence the value of $\partial_S f$ is always at 0 for such edges. In contrary, any edge $\ell \in K_n-\cF$ which adds a cycle to $\cF$ can either belong or not belong to~$S$ in~\eqref{formul1}. Regrouping the corresponding two terms, we obtain a value at~0 plus an integral of a derivative from~0 to the maximal value allowed by the condition $\cF(S, \sigma) =\cF$. But this maximal value is nothing but $X^\cF (w_\cF)$. In this way we have summed all sectors~$\sigma$ inducing a~f\/ixed Hepp sector $\tilde \sigma$ on $\cF$. Finally the remaining sum over $\tilde \sigma$ reconstructs exactly the integration range from~0 to~1 for all edges of~$\cF$.

The fact that $X^\cF (w_\cF)$ is positive for any ordering $\sigma$ now stems from the fact that, for any Hepp sector $\tilde \sigma$ of $\cF$ it is a convex barycentric combination of block matrices
\begin{gather*}
 X^\cF(w_\cF) = (1-w_{\ell_1}) \cId + (w_{\ell_1} -w_{\ell_2}) X^{\Pi_1} + (w_{\ell_2} -w_{\ell_3})X^{\Pi_2} + \dots + w_{\ell_k} X^{\Pi_k} .
\end{gather*}
Remark however that this barycentric combination depends on $\tilde \sigma$. Hence $X^{\cF}$ is in ${\rm PS}_n$ for any~$w_\cF$, as announced, but for a dif\/ferent reason in each dif\/ferent sector (ordering) of the parameters~$w_\cF$.
\end{proof}

We give now a useful corollary of this theorem which expands Gaussian integrals over replicas. Consider indeed a Gaussian measure $d\mu_C$ of covariance $C_{pq}$ on a vector variable $\vec \tau$ with $N$ components $\tau_p$. To study approximate factorization properties of the integral of a product of $n$ functions of the variable $\vec \tau$ it is useful to f\/irst rewrite this integral using a replica trick. It means writing the integral over $n$ identical replicas $\vec \tau_i$ for $i=1, \dots , n$ with components $\tau_{p,i}$, with the perfectly well-def\/ined measure with covariance $[C\otimes \bbone]_{p,i ; q,j} = C_{pq} \bbone_{ij} = C_{pq} $:
\begin{gather*}
 I:= \int d\mu_C (\vec \tau) \prod_{i=1}^n f_i(\vec \tau) = \int d\mu_{C\otimes \bbone} (\vec \tau_i) \prod_{i=1}^n f_i(\vec \tau_i) .
\end{gather*}
Applying the forest formula to this Gaussian integral we obtain the following corollary
\begin{cor} \label{corolgaussfor}
\begin{gather*} I = \sum_{\cF} \int dw_\cF \int d\mu_{C\otimes X^\cF (w_\cF)} (\vec \tau_i)
 \; \partial^C_\cF \prod_{i=1}^n f_i(\vec \tau_i), 
\end{gather*}
where $\partial^C_\cF$ means $\prod\limits_{\ell =(i,j)\in \cF} \Bigl( \sum\limits_{p,q}\frac{\partial}{\partial \tau_{p,i}} C_{pq} \frac{\partial}{\partial \tau_{q,j}} \Bigr)$.
\end{cor}

The proof follows directly from rewriting the Gaussian integral as
\begin{gather*}
\int d\mu_C f (x) = e^{ \frac{\partial}{\partial \tau_i} C_{ij} \frac{\partial}{\partial \tau_j}} f \Big{\vert}_{\tau=0}   .
\end{gather*}
This is the corolary used in many quantum f\/ield theory applications, such as cluster expansions or loop vertex expansions in which the forest formula is used to decouple terms which are coupled through propagators, i.e., through the covariance of a Gaussian measure.

\subsection{Jungle formulas}

The forest formula allows to test links for connecting nodes in a kind of minimal and symmetric way. However in many physical situations, and especially in situations with many scales involved, we would like to introduce a \emph{hierarchy} between the links connecting nodes, and we would like to test f\/irst connectedness through the most interesting links, those of the highest energy/shortest length.

Jungle formulas \cite{Abdesselam:1994ap} provide an abstract solution to this question. They are a simple genera\-li\-za\-tion of forest formulas.

\begin{Definition} Let $m\ge1$ be an integer. An {\sl m-jungle} over the set $V_n$ of $n$ labeled vertices is a sequence $\gF=(\cF_1,\ldots,\cF_m)$ of forests on the complete graph $K_n$ built on $V_n$ such that $\cF_1\subset\cdots\subset\cF_m$.
\end{Definition}

Keeping the same notations than in the previous subsection we now consider a function $f (X^{1}, \dots , X^m)$ def\/ined over $m$ copies of the previous space ${\rm PS}_n$, hence over $[{\rm PS}_n]^m$, which is smooth in the interior with continuous extension of all derivatives to the boundary. We want to Taylor expand with priority to the f\/irst matrix $X^1$, then to $X^2$ and so on. That is we shall test f\/irst how to link the vertices of $V_n$ into larger connected components f\/irst through a forest spanned by the links due to Taylor expansion of the of\/f-diagonal part of~$X_1$ then complete this forest (if it has still several connected components) by Taylor expanding in the of\/f diagonal part of the second of the matrix~$X_2$, but so that the new links together with the f\/irst ones still form a forest on $V_n$ and so on. The result is the $m$-level jungle formula of~\cite{Abdesselam:1994ap}.

\begin{thm}[the jungle formula]
\begin{gather} f( {\bbone })
= \sum_{\gF} \int dw_\gF \, \partial_\gF f  \big[ X^\gF (w_\cF) \big], \label{bkarjungle}
\end{gather}
where
\begin{itemize}\itemsep=0pt
\item the sum over $\gF$ is over $m$-jungles over the set $V_n$ of $n$ labeled vertices $V_n =\{ 1, \dots , n\}$, including the empty jungle $\gF=(\cF_1,\dots,\cF_m)$ with $\cF_i = \varnothing$ $\forall\, i$,

\item $\int dw_\gF \equiv \prod\limits_{\ell\in \cF_m} \int_0^1 dw_\ell $ means integration from~$0$ to~$1$ over one parameter for each edge in the final forest~$\cF_m$ of the jungle,

\item $ \partial_\gF = \prod\limits_{k=1}^m \prod\limits_{\ell\in \cF_k\backslash\cF_{k-1}} \partial_\ell $ means a product of first-order partial derivatives with respect to the variables $X^k_{\ell}$ corresponding to the edges of $ \cF_k\backslash\cF_{k-1}$ $($by convention $\cF_0 = \varnothing)$,

\item $X^\gF (w_\gF)$ is the sequence of $m$ matrices $X^{\gF,k} (w_\gF) $, $k=1, \dots, m$, defined by $X^{\gF,k}_{ii} (w_\gF )= 1$ $\forall\, i$, and for $i \ne j$, $X^{\gF,k}_{ij} (w_\cF)$ is
\begin{itemize}\itemsep=0pt

\item $1$ if $i$ and $j$ are connected by $\cF_{k-1}$,

\item $0$ if $i$ and $j$ are not connected by $\cF_{k}$,

\item the infimum of the $w_\ell$ parameters for $\ell$ in the intersection of $\cF_k\backslash\cF_{k-1}$ and the unique path of $\cF_k$ from~$i$ to~$j$,
when $i$ and $j$ are connected by $\cF_{k}$ but not by $\cF_{k-1}$,
\end{itemize}

\item the symmetric $n$ by $n$ matrices $X^{\gF,k} (w_\gF)$ defined in this way are all positive, hence $X^{\gF} (w_\gF)$ $\in [{\rm PS}_n]^m$, for any value of~$w_\gF$.
\end{itemize}
\end{thm}

The proof is an easy generalization of the proof of the forest formula. Remark that the two-level jungle formula is a main ingredient in multiscale loop vertex expansions~\cite{Gurau:2013oqa}.

\section{Constructive tools}

\subsection{The loop vertex expansion}

The loop vertex expansion (LVE) combines an intermediate f\/ield functional integral representation for QFT quantities with the forest formula and a replica trick similar to the one of the previous section. It allows the computation of connected functional QFT integrals such as the free energy or connected Schwinger functions as convergent sums indexed by spanning trees of arbitrary size $n$ rather than divergent sums indexed by Feynman graphs.

Initially introduced to analyze \emph{matrix} models with quartic interactions~\cite{Rivasseau:2007fr}, the LVE has been extended to analyze random \emph{tensor} models \cite{Delepouve:2014bma,Delepouve:2014hfa,Gurau:2013pca, Magnen:2009at}. In this subsection and the next one we follow~\cite{Gurau:2014vwa}.

Canonical barycentric weights $w(G,T)$ can then be associated to any pair made of a connected graph $G$ and a spanning tree $T\subset G$ by considering the percentage of Hepp sectors in which the tree is leading for Kruskal's greedy algorithm:
\begin{gather*} w(G,T) =\frac{N(G,T)}{|E|!}, 
\end{gather*}
where $N(G,T)$ is the number of sectors $\sigma$ such that $T(\sigma)=T$.

Obviously these weights are barycentric, which simply means that
\begin{gather*} \sum_{ T \subset G} w(G, T) = 1 ,
\end{gather*}
where the sum runs over all spanning trees of $G$. Moreover we have, in the notations of Theorem~\ref{theoforest} the integral representation~\cite{Rivasseau:2013ova}.

\begin{Lemma}
\begin{gather*} w(G,T) = \int dw_T \prod_{\ell \in G-T} X^T_{i(\ell)j(\ell)} (w_T) . 
\end{gather*}
\end{Lemma}

\begin{proof} We introduce f\/irst parameters $w_\ell$ for all the edges in $G-T$, writing
\begin{gather*} X^T_{ij}(\{w \}) = \int_0^1 dw_\ell \bigg[ \prod_{\ell' \in P^T_{i \to j }} \chi(w_\ell < w_{\ell'} ) \bigg] ,
\end{gather*}
where $\chi(\cdots)$ is the characteristic function of the event $\cdots$. Then we decompose the $w$ integrals according to all possible orderings~$\sigma$:
\begin{gather*}
 \int_{0}^1 \prod_{\ell\in G}dw_\ell \prod_{\ell \not\in T} \bigg[ \prod_{\ell' \in P^T_\ell} \chi(w_\ell < w_{\ell'} ) \bigg]
=\sum_{\sigma} \chi ( T(\sigma) =T) \int_{0< w_{\sigma(|E|)} < \cdots < w_{\sigma(1)} < 1} \prod_{\ell\in G}dw_\ell .
\end{gather*}
Indeed in the domain of integration def\/ined by $0< w_{\sigma(|E|)} < \cdots < w_{\sigma(1)} < 1$ the function $\prod\limits_{\ell \not\in T} \bigl[ \prod\limits_{\ell' \in P^T_\ell} \chi(w_\ell < w_{\ell'} ) \bigr]$ is 1 or zero depending whether $ T(\sigma) =T$ or not, as this function being 1 is exactly the condition for Kruskal's algorithm to pick exactly~$T$. Strict inequalities are easier to use here: of course equal values of $w$ factors have zero measure anyway. Hence
\begin{gather*} \int dw_T \prod_{\ell \in G-T} X^T_{i(\ell)j(\ell)} (w_T) =\frac{N(G,T)}{|E|!} = w(G,T).\tag*{\qed} 
\end{gather*}
\renewcommand{\qed}{}
\end{proof}

The LVE expressed any Schwinger function $S$ as a convergent sum over trees of the \emph{intermediate field representation}:
\begin{gather*} S = \sum_{T} A_T, \qquad A_T = \sum_{G \supset T} w(G,T) A_G  ,
\end{gather*}
with
\begin{gather*} \sum_{T} \vert A_T\vert < +\infty .
\end{gather*}
The usual (divergent) perturbative expansion of $S$ is obtained by the ill def\/ined commutation of the sums over~$T$ and~$G$,
\begin{gather*} 
 S = \sum_{T} \bigg(\sum_{G \supset T} w(G,T) A_G \bigg)   \text{``$=$''} \sum_{G } \sum_{T \subset G} w(G,T) A_G = \sum_{G} A_G ,\\
\sum_{G} \vert A_G \vert = \infty  .
\end{gather*}

We shall limit ourselves here to introduce the LVE in the particularly simple case of the quartic $N$-vector models, for which the $1/N$ expansion is governed by rooted plane trees.

More precisely, consider a pair of conjugate vector f\/ields $\{\phi_p \}$, $\{\bar \phi_p \}$, $ p=1 ,\dots , N $, with $ (\bar \phi \cdot \phi )^2$ interaction. The corresponding functional integral
\begin{gather*}
 Z(z, N) = \int \frac{d \bar \phi d \phi }{(2i\pi )^N}  e^{- (\bar \phi \cdot \phi ) + \frac{z}{2N} (\bar \phi \cdot \phi )^2 }
\end{gather*}
is convergent for $\Re z<0$. Note the slightly unusual sign convention for the interaction term. We rewrite it, using a scalar intermediate f\/ield~$\sigma$, as
\begin{gather*} 
Z(z, N) = \int d \sigma \frac{e^{- \sigma^2 /2 } }{\sqrt{2\pi}} \int \frac{d \bar \phi d \phi }{(2i\pi )^N}
 e^{- (\bar \phi \cdot \phi ) + \sqrt {z/N} (\bar \phi \cdot \phi ) \sigma }
= \int \frac{d \sigma }{\sqrt{2\pi}} e^{- \sigma^2 /2 - N \log (1 - \sqrt { z/N} \sigma )} .
\end{gather*}
Def\/ining $\tau = \sigma / \sqrt N $ one gets
\begin{gather} Z(z, N) = \int \frac{\sqrt N d \tau }{\sqrt{2\pi}}e^{- N [\tau^2 /2 + \log (1 - \sqrt {z} \tau )]} . \label{partitfunc}
\end{gather}
The two point function
\begin{gather*} G_2(z, N) = \frac{1}{Z(z, N)} \int \frac{d \bar \phi d \phi }{(2i\pi )^N}  \frac{1}{N} \bigg(\sum_p \bar \phi_p \phi_p\bigg) e^{- (\bar \phi \cdot \phi ) + \frac{z}{2N} (\bar \phi \cdot \phi )^2 } ,
\end{gather*}
can be deduced from the free energy by a Schwinger--Dyson equation
\begin{gather*}
 0 = \frac{1}{Z(z, N)} \int \frac{d \bar \phi d \phi }{(2i\pi )^N}   \frac{1}{N}\sum_p \frac{\partial}{\partial \phi_p}
\big[ \phi_p e^{- (\bar \phi \cdot \phi ) +\frac{z}{2N} (\bar \phi \cdot \phi )^2 } \big] ,
\end{gather*}
which yields
\begin{gather}
G_2(z, N) = 1 + 2 z \frac{d}{d z } \left(\frac{1}{N} \log \int \frac{ \sqrt{N}d \tau }{\sqrt{2\pi}}e^{- N [\tau^2 /2 + \log (1 - \sqrt {z} \tau )]} \right) . \label{2pointint}
\end{gather}

A simple saddle point evaluates the integral \eqref{partitfunc} as $\frac{K e^{- N f(\tau_c) } }{ \sqrt{ f" (\tau_c)}}$, where the saddle point of $f(\tau) = \tau^2 /2 + \log (1 - \sqrt {z} \tau )]$ is at $\tau_c$ with $f'( \tau_c)=0$ hence $\tau_c = \frac{1}{{2 \sqrt z}} [1 - \sqrt{1 - 4 z}]$. Also
\begin{gather*}
\lim_{N \to \infty} \frac{ \log Z(z, N)}{N} = - f(\tau_c) ,
\end{gather*}
and the two point function in the $N\to \infty$ limit is{\samepage
\begin{gather} \lim_{N \to \infty} G_2(z, N)
= 1 + 2 z \left(- \partial_z f(\tau_c) -\partial_{\tau } f(\tau_c) \frac{d\tau_c}{dz} \right) =
 1 - 2z \frac{ -\frac{1}{2 \sqrt{z} } \tau_c }{ 1-\sqrt{z}\tau_c }\nonumber\\
\hphantom{\lim_{N \to \infty} G_2(z, N)}{}
= \frac{1}{2z} \big[ 1 - \sqrt{1 - 4 z} \big]  , \label{catalansigma}
\end{gather}
which we recognize as the generating function of the Catalan numbers.}

Let us now study Borel summability in $z$ of these quantities uniformly as $N \to \infty$, using the loop vertex expansion. We start from the intermediate f\/ield representation of the two-point function~\eqref{2pointint} and apply the LVE to get
\begin{gather} \label{vectorlve}
G_2(z, N) = \sum_{\cT} \frac{1}{n!} z^{n} \int dw_\cT \int d\mu_\cT \prod_{c \in C(\cT)} \frac{1}{1- \sqrt {z} \tau_{i(c)}},
\end{gather}
where in \eqref{vectorlve}
\begin{itemize}\itemsep=0pt

\item the sum over $\cT$ is over rooted plane trees, with one ciliated root vertex labeled $i=0$ plus $n \ge 0$ ordinary vertices labeled $1, \dots ,n$,

\item $\int dw_\cT$ as in Subsection~\ref{forformul} means $\big[ \prod\limits_{\ell\in \cT} \int_0^1 dw_\ell \big]$,

\item $d\mu_\cT $ is the normalized Gaussian measure on the $(n+1)$-dimensional vector f\/ield $\vec \tau = (\tau_i)$, $i=0,1, \dots, n$, running over the vertices of $\cT$, which has covariance $\frac{X_{ij}^\cT (w_\cT)}{N}$ between vertices $i$ and $j$. Recall that $X^\cT (w_\cT)$ is def\/ined in Subsection~\ref{forformul},

\item the product over $c$ runs over the set $C(\cT)$ of the $2n+1$ corners of the tree, the cilium creating an additional corner on the plane tree, and $i(c)$ is the index of the vertex to which the corner $c$ belongs.
\end{itemize}

It is now obvious why \eqref{catalansigma} is true; since the covariance of the $\tau$ f\/ields vanishes as $N \to \infty$ the limit of $G_2(z, N) $ is obtained by putting every $\tau_{i(c)}$ factor to 0 in every corner resolvent, in which case we exactly get a weight $z^{n(\cT)}$ for each rooted plane tree, hence we recover the Catalan generating function (the $1/n!$ is canceled by relabeling of the vertices).

We can now use \eqref{vectorlve} to prove analyticity and Borel summability of the free energy and correlation functions of the model in the variables~$z$ and $1/N$ in the cardioid domain of Fig.~\ref{cardioidfig} (see~\cite{Frohlich:1982rj} for an early reference to Borel summability of the $1/N$ expansion of vector models).

Let us set $z = \vert z \vert e^{i \pi + i\phi}$ for $\vert \phi \vert <\pi $. We have $\sqrt z = i\sqrt{ \vert z \vert} e^{i \phi/2}$. Each resolvent $\frac{1}{1- \sqrt {z} \tau_{i(c)}}$ is bounded in norm by $[\cos (\phi/2)]^{-1}$, hence using the fact that there are $2n+1$ such resolvents, we obtain analyticity of representation \eqref{vectorlve} for $4 \vert z \vert < [ \cos (\phi/2 )]^{2} $, the cardioid domain of Fig.~\ref{cardioidfig}.

But in fact we can extend the analyticity domain into the extended cardioid domain of Fig.~\ref{extendedcardioidfig}, a~domain introduced for quartic vector models in~\cite{Billionnet:1982cd}. Indeed using the parametric representation of resolvents
\begin{gather*} 
\frac{1}{1- \sqrt {z} \tau} = \int_0^\infty d \alpha_c e^{- \alpha_c (1 - \sqrt {z} \tau ) }
\end{gather*}
we can explicitly integrate over the measure $d\mu_{\cT}$ and get the integral representation
\begin{gather*}
 G_2(z, N)  =  \sum_{\cT} \frac{1}{n!} z^{n} \int dw_\cT \left[ \prod_{c\in \cT} \int_0^\infty d \alpha_c e^{- \alpha_c } \right]
 e^{ \frac{z}{2N} \sum\limits_{ij}\big(\sum\limits_{c \in i} \alpha_c\big) X^\cT_{ij}(w_\cT ) \big(\sum\limits_{c \in j} \alpha_c\big) }  .
\end{gather*}
The formula above can be further simplif\/ied. Putting $\beta_i = \sum_{c \in i} \alpha_c $ we have
\begin{gather}
 G_2(z, N)  =  \sum_{\cT} \frac{1}{n!} z^{n} \left[ \prod_{i =0}^n \int_0^\infty d \beta_i \frac{\beta_i^{d_i-1}}{(d_i-1)!} e^{- \beta_i} \right]
 \int dw_\cT e^{ \frac{z}{2N} \sum\limits_{ij} \beta_i X^\cT_{ij}(w_\cT ) \beta_{j} }, \label{betarep}
\end{gather}
where $d_i$ is the degree of $i$, hence the number of corners of $i$.

\begin{figure}[t]\centering
\includegraphics[width=7cm]{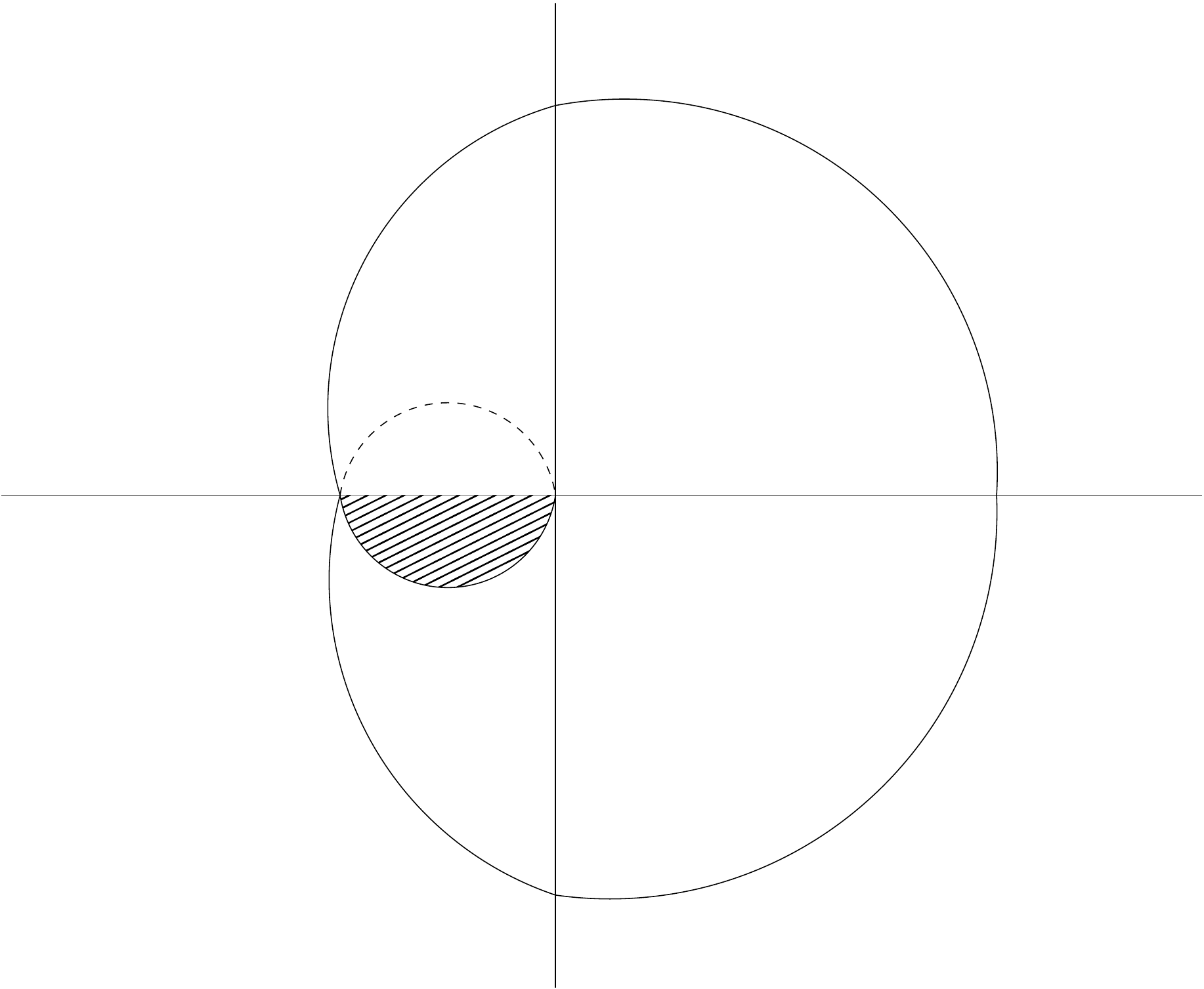}
\caption{The extended cardioid domain of vector models. The crosshatching represents the analytic continuation above the negative axis cut; there is a~symmetric analytic continuation below the cut.}\label{extendedcardioidfig}
\end{figure}

Setting $z = \vert z \vert e^{i \pi + i\phi}$ and $\beta = \vert \beta \vert e^{i \psi}$ we have, for $ - \pi/2 \le \phi + 2 \psi \le \pi/2$ and $- \pi/2 < \psi < \pi /2$, both $\cos ( \phi + 2 \psi ) \ge 0 $
and $\cos \psi > 0 $, hence
\begin{gather*}
\left\vert \left[ \prod_{i =0}^n \int_0^{e^{i\psi}\infty} d \beta_i \frac{\beta_i^{d_i-1}}{(d_i-1)!} e^{- \beta_i} \right]
e^{ \frac{z}{2N} \sum\limits_{ij} \beta_i X^\cT_{ij}(w_\cT ) \beta_j } \right\vert \\
\qquad{} \le \prod_{i =0}^n \int_0^\infty d \vert \beta_i \vert \frac{\vert \beta_i\vert ^{d_i-1}}{(d_i-1)!} e^{-\vert \beta_i \vert \cos \psi}
\le ( \cos \psi )^{-\sum\limits_{i =0}^n d_i } = ( \cos \psi )^{-2n-1} .
\end{gather*}
Therefore $G_2(z, N)$ is analytic in the extended cardioid $C = C_+ \cup C_-$, where $C_+$ is the union of the quarter-disk $0 \le \phi < \pi /2$, $ 4\vert z \vert < 1 $ and of the domain $\pi/2 \le \phi < 3\pi /2$, $ 4\vert z \vert < [ \cos (\phi/2 -\pi/4 )]^{2} $. $C_-$ is the complex conjugate domain.

{\allowdisplaybreaks
To prove that this convergent analytic function is the Borel sum of its perturbative series at any f\/ixed $N$ also requires uniform Taylor estimates
of the type $K^p p! \vert z\vert^{p}$ for the Taylor remainder at order $p$ in at least a disk tangent to the imaginary axis (Nevanlinna's criterion).
They follow from Taylor expanding the exponential of the $\beta$ quadratic form with an integral remainder:
\begin{gather*}  e^{ \frac{z}{2N} \sum\limits_{ij} \beta_i X^\cT_{ij}(w_\cT ) \beta_j } =
\sum_{q=0}^{p-1} \frac{z^q}{q!(2N)^q} \bigg[ \sum_{ij} \beta_i x^\cT_{ij}(w_\cT ) \beta_j \bigg]^q \\
\hphantom{e^{ \frac{z}{2N} \sum\limits_{ij} \beta_i X^\cT_{ij}(w_\cT ) \beta_j } =}{}  + \int_0^1 dt \frac{ (1-t)^{p-1}}{(p-1)!}
\frac{z^{p}}{(2N)^{p} }\bigg[ \sum_{ij} \beta_i X^\cT_{ij}(w_\cT ) \beta_j \bigg]^{p} e^{ t \frac{z}{2N} \sum\limits_{ij} \beta_i X^\cT_{ij}(w_\cT ) \beta_j } .
\end{gather*}
The sum over $q$, i.e., the $p$ f\/irst terms, are exactly the perturbative expansion up to order~$p$ hence support a $K^pp!|z|^p$ bound. The Taylor remainder term for any tree $\cT$ in the disk $ - \pi/2 \le \phi \le \pi /2$, where we can take $\psi =0$ can be bounded as
\begin{gather*}
\left \vert  \left[ \prod_{i =0}^n \int_0^{\infty} \frac{\beta_i^{d_i-1}}{(d_i-1)!} d \beta_i e^{- \beta_i} \right] \int_0^1 dt \frac{ (1-t)^p}{p!}
\frac{z^{p}}{(2N)^{p} } \right. \\
\left.\qquad\quad{}\times
\int dw_\cT \bigg[ \sum_{ij} \beta_i X^\cT_{ij}(w_\cT ) \beta_j \bigg]^{p} e^{ t \frac{z}{2N} \sum\limits_{ij} \beta_i X^\cT_{ij}(w_\cT ) \beta_j } \right\vert \\
\qquad{} \le \frac{\vert z \vert ^{p}}{p!}\left[ \prod_{i =0}^n \int_0^{\infty} \frac{\beta_i^{d_i-1}}{(d_i-1)!} d \beta_i e^{- \beta_i} \right]
\left[\sum_{i=0}^n \beta_i \right]^{2p}  \\
\qquad{}  = \frac{\vert z \vert ^{p}}{p!} \left( \prod_{i=0}^{n} \frac{1}{(d_i-1)!} \right) \int_0^{\infty} d\beta e^{-\beta} \beta^{2p }
\int_{\beta_1+\dots +\beta_n=\beta } \prod_{i=0}^{n} \beta_i^{d_i-1} d\beta_i\\
\qquad{} = \frac{\vert z \vert ^{p}}{p!} \left( \prod_{i=0}^{n} \frac{1}{(d_i-1)!} \right) (2p+2n+1)!
\int_{u_1+\dots +u_n=1 } \prod_{i=0}^{n} u_i^{d_i-1} du_i \le 4^n K^p p!|z|^p  .
\end{gather*}
These Taylor estimates for a single rooted plane tree can be summed over all rooted plane trees (using the $\vert z\vert ^n$ factor in~\eqref{betarep})
in the half-disk def\/ined by $16 \vert z \vert <1$ and $ - \pi/2 \le \arg z \le \pi /2$ (shown in red on Fig.~\ref{extendedcardioidfig}). Hence in this half-disk (which is uniform in $N$) we obtain the desired Taylor estimates, which is more than enough to check that the expansions \eqref{vectorlve} and~\eqref{betarep} represent indeed for all $N$ the unique Borel sum of the perturbative series.

}

Interesting functions are the real and imaginary parts along the real axis $0\le z < 1/8$ which are
\begin{gather*} G^{\rm mean}_{2} (z, N)= \frac{ G_2(z, N)_+ + G_2(z, N)_-}{2}, \qquad G^{\rm cut}_{2} (z, N)= \frac{ G_2(z, N)_+ - G_2(z, N)_-}{2i},
\end{gather*}
where $G_+$ is analytically continued to $\phi = +\pi$ and $G_-$ is analytically continued to $\phi = -\pi$. Taking $\psi = -\pi /4$ in the f\/irst case and $\psi = + \pi/4$ is the second case, one obtains explicitly convergent integral representations for these quantities, namely
\begin{gather*} 
G^{\rm mean}_{2} (z, N) = \sum_{\cT} \frac{1}{n!} z^{n}
\left[ \prod_{i \in V(\cT)} \int_0^\infty \frac{\beta_i^{d_i-1}}{(d_i-1)!} d \beta_i e^{- \frac{\sqrt 2}{2} \beta_i} \right] \\
\hphantom{G^{\rm mean}_{2} (z, N) =}{}
 \times \int dw_\cT \cos \left((2n+1) \frac{\pi}{4} + \frac{\sqrt 2}{2} \sum_i \beta_i + \frac{z}{2N} \sum_{ij} \beta_i X^\cT_{ij}(w_\cT ) \beta_{j} \right) ,\\ 
 G^{\rm cut}_{2} (z, N) = \sum_{\cT} \frac{1}{n!} z^{n} \int dw_\cT
\left[ \prod_{i \in V(\cT)} \int_0^\infty \frac{\beta_i^{d_i-1}}{(d_i-1)!} d \beta_i e^{- \frac{\sqrt 2}{2} \beta_i} \right] \\
\hphantom{G^{\rm cut}_{2} (z, N) = }{}
 \times \int dw_\cT \sin \left((2n+1) \frac{\pi}{4} + \frac{\sqrt 2}{2} \sum_i \beta_i + \frac{z}{2N} \sum_{ij} \beta_i X^\cT_{ij}(w_\cT ) \beta_{j} \right) . \end{gather*}
where the factors $(2n+1) \frac{\pi}{4}$ come from the rotation of the $\beta$ integrals, using $\sum d_i = 2n+1$. These convergent integrals extend half-way to the Catalan singularity $ z_{\text{Catalan}}=1/4 $. Indeed bounding the cosine or sinus function by~1 we obtain convergence, but loosing a factor $(\sqrt 2)^{\sum d_i} = 2^n \sqrt{2}$.

One can still check easily that the limit for $N \to \infty$ of the mean integral for positive $z$ is the Catalan function. Indeed the cosine function simplif\/ies in that case. Rotating the $\beta$ integrals back in position we obtain again the factor 1 for each rooted plane tree.

The extended cardioid is an analyticity domain in $z$ which holds for any $N \ge 1$. In other words it is common to all $N$-vector models, including the particular $N=1$ scalar case, However as $N \to \infty$ we could hope for larger and larger domains of analyticity which approach the $z = 1/4$ singularity when $N\to \infty $; but we do not know, even in this simple vector model case, how to prove this.

In the case of quartic large $N$ matrix \cite{Rivasseau:2007fr} and large $N$ tensor models \cite{Delepouve:2014bma,Delepouve:2014hfa,Gurau:2013pca, Magnen:2009at}, the LVE also provides analyticity in cardioid-like domains.

The constructive treatment of renormalizable models requires a~multiscale analysis, hence a~mul\-ti\-scale version of the loop vertex expansion (MLVE). Following \cite{Gurau:2013oqa}, we sketch now how this expansion works in the case of a super-renormalizable toy model which is a slight modif\/ication of the vector model above.

\subsection{Multiscale loop vertex expansion}

Consider the same pair of conjugate vector f\/ields $\{\phi_p \}$, $\{\bar \phi_p \}$, $p=1 ,\dots , N $, with the same \mbox{$\frac{\lambda^2 }{2}(\bar \phi \cdot \phi )^2$} bare interaction as in the previous section, but with a dif\/ferent Gaussian measure $d\mu (\bar \phi, \phi)$ which breaks the ${\rm U}(N)$ invariance of the theory. It has diagonal covariance (or propagator) which decreases as the inverse power of the f\/ield index:
\begin{gather*}
 d\eta(\bar \phi, \phi) = \left( \prod_{p=1}^N p \frac{ d\bar \phi_p d\phi_p }{2\pi \imath} \right) e^{-\sum_{p=1}^N p  \bar \phi_p \phi_p }  ,
\qquad \int d\eta (\bar \phi, \phi)  \bar\phi_p \phi_{q}= \frac{\delta_{pq}}{p}.
\end{gather*}
This propagator renders the perturbative amplitudes of the model f\/inite in the $N \to \infty$ limit, except for a mild divergence of self-loops which yields a logarithmically divergent sum $L_N = \sum\limits_{p=1}^N \frac{1}{p} \simeq \log N$. These divergences are easily renormalized by using a vector-Wick-ordered $\phi^4$ interaction, namely $\frac{1}{2}[\lambda (\bar \phi \cdot \phi -L_N)]^2$. Remark that this interaction (contrary to the $\phi^4_2$ case) remains positive for $\lambda$ real at all values of $(\bar\phi, \phi)$. The renormalized partition function of the model is
\begin{gather*}
Z(\lambda, N) = \int d\eta (\bar \phi, \phi ) \, e^{- \frac{\lambda^2}{2} (\bar \phi \cdot \phi -L_N)^2 }.
\end{gather*}
The intermediate f\/ield representation decomposes the quartic interaction using an intermediate scalar f\/ield $ \sigma $:
\begin{gather*}
e^{- \frac{\lambda^2}{2} (\bar \phi \cdot \phi -L_N)^2 } = \int d\nu (\sigma) \,
e^{ \imath \lambda \sigma (\bar \phi \cdot \phi -L_N) } ,
\end{gather*}
where $d\nu(\sigma) = \frac{1}{\sqrt{2\pi}} e^{-\frac{\sigma^2}{2} }$ is the standard Gaussian measure with covariance~1. Integrating over the initial f\/ields $(\bar \phi_p, \phi_p)$ leads to
\begin{gather*}
Z(\lambda, N) = \int d\nu (\sigma) \,  \prod_{p=1}^N \frac{1}{1- \imath \frac{\lambda \sigma}{p}} e^{-\imath \frac{\lambda \sigma}{p} }
 = \int d\nu (\sigma) \, e^{- \sum\limits_{p=1}^N \log_2 \bigl(1 - \imath \frac{\lambda \sigma }{p } \bigr) }  ,
\end{gather*}
where $\log_2 (1-x) \equiv x+ \log (1-x) = O(x^2)$.

Applying the ordinary LVE of the previous section to this functional integral would express $\log Z(\lambda, N)$ as a sum over trees, but there is no simple way to remove the logarithmic divergence of all leaves of the tree without generating many intermediate f\/ields in numerators which, when integrated through the Gaussian measure, would create an apparent divergence of the series. The MLVE is designed to solve this problem.

We f\/ix an integer $M>1$ and def\/ine the $j$-th \emph{slice}, as made of the indices $p \in I_j \equiv [M^{j-1},M^{j} -1]$. The ultraviolet cutof\/f $N$ is chosen as $N =M^{j_{\max}}-1$, with $j_{\max}$ an integer. We can also f\/ix an infrared cutof\/f $j_{\min}$. Hence there are $j_{\max}-j_{\min}$ slices in the theory, and the ultraviolet limit corresponds to the limit $j_{\max} \to \infty$. The intermediate f\/ield representation writes:
\begin{gather*}
Z(\lambda, N) = \int d\nu (\sigma)   \prod_{j =j_{\min}}^{j_{\max}} e^{- V_j} , \qquad
V_{j} =\sum_{p \in I_{j}} \log_2 \left(1 - \imath \frac{ \lambda \sigma}{p} \right)  .
\end{gather*}
The factorization of the interaction over the set of slices $\cS = [j_{\min}, \dots, j_{\max}]$ can be encoded into an integral over Grassmann numbers. Indeed,
\begin{gather*}
 a = \int d\bar \chi d\chi \, e^{-\bar \chi a \chi} = \int d\mu (\bar \chi ,\chi ) \, e^{- \bar \chi (a-1) \chi},
\end{gather*}
where $d \mu(\bar \chi ,\chi ) = d\bar \chi d\chi \, e^{-\bar \chi \chi}$ is the standard normalized Grassmann Gaussian measure with covariance~1. Hence,
denoting $W_j(\sigma) = e^{-V_{j}}-1$,
\begin{gather*}
Z(\lambda, N) = \int d\nu (\sigma)   \left( \prod_{j = j_{\min}}^{j_{\max} } d\mu (\bar \chi_j , \chi_j) \right)
 e^{ - \sum\limits_{j = j_{\min}}^{j_{\max}} \bar \chi_j W_j ( \sigma) \chi_j } .
\end{gather*}

We now rewrite the partition function as
\begin{gather*}
 Z(\lambda, N) = \int d\nu_\cS \, e^{- W},
\qquad\! d\nu_{\cS} = d\nu_{\bbone_\cS} ( \{ \sigma_j\} ) \, d\mu_{\mathbb{I}_\cS} (\{\bar \chi_j , \chi_j\}),\qquad\!
W = \sum_{j =j_{\min}}^{j_{\max}} \bar \chi_j W_j ( \sigma_j) \chi_j .
\end{gather*}
This is the starting point for the MLVE. The f\/irst step is to expand to inf\/inity the exponential of the interaction
\begin{gather*}
Z(\lambda, N) = \sum_{n=0}^\infty \frac{1}{n!}\int d\nu_{\cS} \, (-W)^n = \sum_{n=0}^\infty \frac{1}{n!} \int d\nu_{\cS,V} \, \prod_{a=1}^n (-W_a) ,
\end{gather*}
where the $a$-th vertex is
\begin{gather*}
W_a = \sum_{j =j_{\min}}^{j_{\max}} W_{a,j} , \qquad W_{a,j} = \bar \chi^a_j W_j ( \sigma^a_j ) \chi^a_j ,
\end{gather*}
and has now its own (replicated) bosonic and fermionic f\/ields $\sigma_j^a$, $\bar \chi^a_j$, $\chi^a_j$ and the replica measure is completely degenerate:
\begin{gather*}
d\nu_{\cS,V} = d\nu_{\bbone_\cS \otimes \bbone_V} (\{ \sigma^a_j\}) \, d\mu_{\mathbb{I}_\cS \otimes \bbone_V} (\{\bar \chi^a_j , \chi^a_j\}) .
\end{gather*}

The obstacle to factorize this integral over vertices lies now in the bosonic and fermionic degenerate blocks $\bbone_V$. In order to deal with these couplings we apply the $m=2$ jungle formula~\eqref{bkarjungle} with priority to the bosonic links. It means that in the measure $d \nu$ one introduces f\/irst the matrix $x_V$ with coupling parameters $x_{ab}=x_{ba}$, $x_{aa}=1$ between the vertex bosonic replicas
\begin{gather*}
 Z(\lambda, N) = \sum_{n=0}^\infty \frac{1}{n!}
 \int d\nu_{\bbone_\cS \otimes x_V} (\{ \sigma^a_j\}) \, d\mu_{\mathbb{I}_\cS \otimes \bbone_V} (\{\bar \chi^a_j , \chi^a_j\})
 \prod_{a=1}^n \left( - \sum_{j =j_{\min}}^{j_{\max}} W_{a,j} \right)
\Big\vert_{x_{ab}=1 }
\end{gather*}
and apply the forest formula. We denote $\cF_B$ a bosonic forest with $n$ vertices labelled $\{1,\dots, n\}$, $\ell_{B}$ a generic edge of the forest and $a(\ell_B), b(\ell_B)$ the end vertices of $\ell_B$. The result of the f\/irst forest formula is
\begin{gather*}
 Z(\lambda, N) = \sum_{n=0}^\infty \frac{1}{n!}
 \sum_{\cF_{B}} \int dw_{\cF_B} \int d\nu_{\bbone_\cS \otimes X(w_{\ell_B})}
 (\{ \sigma^a_j\}) \, d\mu_{\mathbb{I}_\cS \otimes \bbone_V} (\{\bar \chi^a_j , \chi^a_j\})\\
\hphantom{Z(\lambda, N) =}{} \times \partial_{\cF_B}  \prod_{a=1}^n \left( - \sum_{j =j_{\min}}^{j_{\max}} W_{a,j} \right),
\end{gather*}
where
\begin{gather*}
 \int dw_{\cF_B} = \prod_{\ell_B\in \cF_B } \int_{0}^1 dw_{\ell_B} , \qquad
\partial_{\cF_B}= \prod_{\ell_B \in \cF_{B}} \left( \sum_{j,k=j_{\min}}^{j_{\max}}
 \frac{\partial}{\partial \sigma^{a(\ell_B)}_j}\frac{\partial}{\partial \sigma^{b(\ell_B)}_k} \right),
 \end{gather*}
and $X_{ab}(w_{\ell_B})$ is the inf\/imum over the parameters $w_{\ell_B}$ in the unique path in the forest~$\cF_B$ connecting $a$ and $b$, and the inf\/imum is set to~$1$ if $a=b$ and to zero if $a$ and $b$ are not connected by the forest.

The forest $\cF_B$ partitions the set of vertices into blocks $\cB$ corresponding to its trees. Remark that the blocks can be singletons (corresponding to the trees with no edges in~$\cF_B$). We denote $a\in \cB$ if the vertex $a$ belongs to a~bosonic block~$\cB$. A~vertex belongs to a unique bosonic block. Contracting every bosonic block to an ``ef\/fective vertex'' we obtain a graph which we denote~$\{n\}/\cF_B$. We introduce decoupling parameters $y_{\cB\cB'}=y_{\cB'\cB}$ for the fermionic f\/ields $\chi^{\cB}_j$ corresponding to the blocks of $\cF_B$ (i.e., for the ef\/fective vertices of~$\{n\}/\cF_B$). Applying (a second time) the forest formula, this time for the $y$'s, is exactly the level~2 jungle formula as it leads to a set of fermionic edges $\cL_F$ forming a forest in $\{n\}/\cF_B$ (hence connecting bosonic blocks). We denote $L_{F} $ a generic fermionic edge connecting blocks and $\cB(L_F), \cB'(L_F) $ the end blocks of the fermionic edge~$L_F$. We obtain
\begin{gather}
Z(\lambda, N) = \sum_{n=0}^\infty \frac{1}{n!} \sum_{\cF_{B}} \sum_{\cL_F} \int dw_{\cF_B} \int dw_{\cL_F}
\int d\nu_{\bbone_\cS \otimes X(w_{\ell_B})} (\{ \sigma^a_j\}) \nonumber\\
\hphantom{Z(\lambda, N) =}{}  \times d\mu_{\mathbb{I}_\cS \otimes Y ( w_{L_F}) } (\{\bar \chi^\cB_j , \chi^\cB_j\})
\partial_{\cF_B} \partial_{\cL_F}
 \prod_{\cB} \prod_{a\in \cB} \left( - \sum_{j =j_{\min}}^{j_{\max }} \bar \chi^{\cB }_j W_j ( \sigma^a_j )
\chi^{\cB }_j \right), \label{eq:intermediate}
\end{gather}
where
\begin{gather*} \int dw_{\cL_F} = \prod_{L_F\in \cL_F } \int_0^1 dw_{L_F},\\
\partial_{\cL_F}= \prod_{L_F \in \cL_F}
\left(\sum_{j=j_{\min} }^{j_{\max}} \left( \frac{\partial}{\partial \bar \chi_j^{\cB(L_F)} } \frac{\partial}{\partial \chi_j^{\cB'(L_F)} }
+ \frac{\partial}{\partial \bar \chi_j^{\cB'(L_F)} } \frac{\partial}{\partial \chi_j^{\cB(L_F)}} \right)\right)  ,
\end{gather*}
and
$Y_{\cB\cB'}(w_{\ell_F}) $ is the inf\/imum over $w_{\ell_F}$ in the unique path in $\cL_{F}$ connecting $\cB$ and $\cB'$, this inf\/imum being set to $1$ if $\cB= \cB'$ and to zero if $\cB$ and $\cB'$ are not connected by $\cL_F$. Note that the fermionic edges are oriented. Expanding the sums over $j$ in the last line of equation~\eqref{eq:intermediate} we obtain a sum over slice assignments $J= \{j_a \}$ to the vertices~$a$, where $j_a \in [ j_{\min} , j_{\max}]$. Taking into account that $\partial_{\sigma^a_j} W(\sigma^{a}_{j_a}) = \delta_{jj_a} \partial_{\sigma^a_{j_a} } W(\sigma^{a}_{j_a}) $
we obtain:
\begin{gather*}
Z(\lambda, N) = \sum_{n=0}^\infty \frac{1}{n!} \sum_{\cF_{B}} \sum_{\cL_F} \sum_J  \int dw_{\cF_B} \int dw_{\cL_F}\\
\hphantom{Z(\lambda, N) =}{}\times
\int d\nu_{\bbone_\cS \otimes X(w_{\ell_B})} (\{ \sigma^a_j\})
 d\mu_{\mathbb{I}_\cS \otimes Y ( w_{L_F}) } \big(\big\{\bar \chi^\cB_j , \chi^\cB_j\big\}\big) \\
\hphantom{Z(\lambda, N) =}{}\times
 \partial_{\cF_B} \partial_{\cL_F} \prod_{\cB} \prod_{a\in \cB} \big( {-} \bar \chi^{\cB }_{j_a} W_{j_a} ( \sigma^a_{j_a} )
\chi^{\cB }_{j_a} \big) .
\end{gather*}
In order to compute the derivatives in $\partial_{\cL_F} $ with respect to the block fermionic f\/ields $\chi^{\cB}_j$ and $\bar \chi^{\cB}_j$ we note that such a derivative acts only on $ \prod\limits_{a\in \cB} \bigl( \chi^{\cB}_{j_a} \bar \chi^{ \cB }_{j_a} \bigr) $ and, furthermore,
\begin{gather*}
 \frac{\partial}{\partial \bar \chi_j^{\cB} }
 \prod_{a\in \cB} \bigl( \chi^{\cB}_{j_a} \bar \chi^{ \cB }_{j_a} \bigr)
= \left( \sum_{a'\in \cB} \delta_{jj_{a'}} \frac{\partial}{\partial \bar \chi_{j_{a'}}^{\cB} } \right)
 \prod_{a\in \cB} \big( \chi^{\cB}_{j_a} \bar \chi^{ \cB }_{j_a} \big),\\
 \frac{\partial}{\partial \chi_j^{\cB} }
 \prod_{a\in \cB} \bigl( \chi^{\cB}_{j_a} \bar \chi^{ \cB }_{j_a} \bigr)
 = \left( \sum_{a'\in \cB} \delta_{jj_{a'}} \frac{\partial}{\partial \chi_{j_{a'}}^{\cB} } \right)
 \prod_{a\in \cB} \big( \chi^{\cB}_{j_a} \bar \chi^{ \cB }_{j_a} \big).
\end{gather*}
It follows that the Grassmann Gaussian integral is
\begin{gather*}
  \Bigg[ e^{
 \sum\limits_{\cB,\cB'} Y_{\cB\cB'}(w_{\ell_F})\sum\limits_{a\in \cB, b\in \cB'} \delta_{j_aj_b}
 \frac{\partial}{\partial \bar \chi_{j_a}^{\cB} } \frac{\partial}{\partial \chi_{j_b}^{\cB'} } } \\
 \qquad{}\times
  \prod_{L_F \in \cL_F}
 \left(\sum_{a\in \cB(L_F),b\in \cB'(L_{F})} \delta_{j_a j_b}\left( \frac{\partial}{\partial \bar \chi_{j_a}^{\cB(L_F)} }
 \frac{\partial}{\partial \chi_{j_b}^{\cB'(L_F)} } +
 \frac{\partial}{\partial \bar \chi_{j_b}^{\cB'(L_F)} } \frac{\partial}{\partial \chi_{j_a}^{\cB(L_F)} } \right)
 \right) \\
  \qquad {}\times \prod_{\cB} \prod_{a\in \cB} \bigl( \chi^{\cB}_{j_a} \bar \chi^{ \cB }_{j_a} \bigr) \Bigg]_{\chi^{\cB}_j, \bar\chi^{\cB}_j =0 }  .
\end{gather*}
The sums over $ a\in \cB(\ell_F) $ and $ b\in \cB'(\ell_F)$ yield a sum over all the possible ways to hook the edge $L_F\in \cL_F$ to vertices in its end blocks. Each term represents a detailed fermionic edge $\ell_F$ in the original graph (having the same $w_{\ell_F}= w_{L_F}$ parameter). The sum over $\cL_F$ becomes therefore a sum over detailed fermionic forests $\cF_F$ in the original graph (in which the bosonic blocks are not contracted)and we obtain
\begin{gather*}
Z(\lambda, N) = \sum_{n=0}^\infty \frac{1}{n!} \sum_{\cJ} \sum_J
\int dw_\cJ \int d\nu_{ \cJ} \partial_\cJ \prod_{\cB} \prod_{a\in \cB} \bigl( W_{j_a} ( \sigma^a_{j_a} )
\chi^{ \cB }_{j_a} \bar \chi^{\cB}_{j_a} \bigr) ,
\end{gather*}
where
\begin{itemize}\itemsep=0pt
\item the sum over $J$ means $\sum\limits_{j_1=j_{\min}}^{j_{\max } } \cdots \sum\limits_{j_n=j_{\min}}^{j_{\max} }$,

\item the sum over $\cJ$ runs over all two level jungles, hence over all ordered pairs $\cJ = (\cF_B, \cF_F)$ of two (each possibly empty) disjoint forests on $V$, such that $\bar \cJ = \cF_B \cup \cF_F $ is still a forest on $V$. The forests $\cF_B$ and $\cF_F$ are the bosonic and fermionic components of~$\cJ$. The edges of $\cJ$ are partitioned into bosonic edges $\ell_B$ and fermionic edges~$\ell_F$,

\item $\int dw_\cJ$ means integration from 0 to 1 over parameters $w_\ell$, one for each edge $\ell \in \bar\cJ$. $\int dw_\cJ = \prod\limits_{\ell\in \bar \cJ} \int_0^1 dw_\ell $. A generic integration point $w_\cJ$
is therefore made of $\vert \bar \cJ \vert$ parameters $w_\ell \in [0,1]$, one for each $\ell \in \bar \cJ$,

\item
\begin{gather*}
\partial_\cJ = \prod_{\genfrac{}{}{0pt}{}{\ell_B \in \cF_B}{\ell_B=(c,d)}} \left( \frac{\partial}{\partial \sigma^{c}_{j_{ c}} }
 \frac{\partial}{\partial \sigma^{d}_{j_{d } } } \right)
\prod_{\genfrac{}{}{0pt}{}{\ell_F \in \cF_F}{\ell_F=(a,b) } } \delta_{j_{a } j_{b } } \left(
 \frac{\partial}{\partial \bar \chi^{\cB(a)}_{j_{a} } }\frac{\partial}{\partial \chi^{\cB(b)}_{j_{b } } }+
 \frac{\partial}{\partial \bar \chi^{ \cB( b) }_{j_{b} } } \frac{\partial}{\partial \chi^{\cB(a) }_{j_{a} } }
 \right) ,
\end{gather*}
where $ \cB(a)$ denotes the bosonic blocks to which $a$ belongs,

\item the measure $d\nu_{\cJ}$ has covariance $ X (w_{\ell_B}) \otimes \bbone_\cS $ on bosonic variables and $ Y (w_{\ell_F}) \otimes \mathbb{I}_\cS $
on fermionic variables, hence $\int d\nu_{\cJ} f$ is the value at $\sigma = \bar \chi = \chi =0$ of
\begin{gather*}
 e^{\frac{1}{2} \sum\limits_{a,b=1}^n X_{ab}(w_{\ell_B }) \frac{\partial}{\partial \sigma^a_{j_a} }\frac{\partial}{\partial \sigma^b_{j_b}}
 + \sum\limits_{\cB,\cB'} Y_{\cB\cB'}(w_{\ell_F})\sum\limits_{a\in \cB, b\in \cB' } \delta_{j_aj_b}
 \frac{\partial}{\partial \bar \chi_{j_a}^{\cB} } \frac{\partial}{\partial \chi_{j_b}^{\cB'} } }  f ,
\end{gather*}

\item $X_{ab} (w_{\ell_B} )$ is the inf\/imum of the $w_{\ell_B}$ parameters for all the bosonic edges $\ell_B$ in the unique path $P^{\cF_B}_{a \to b}$ from $a$ to $b$ in $\cF_B$. The inf\/imum is set to zero if such a path does not exists and to $1$ if $a=b$,

\item $Y_{\cB\cB'}(w_{\ell_F})$ is the inf\/imum of the $w_{\ell_F}$ parameters for all the fermionic edges $\ell_F$ in any of the paths $P^{\cF_B \cup \cF_F}_{a\to b}$ from some vertex $a\in \cB$ to some vertex $b\in \cB'$. The inf\/imum is set to~$0$ if there are no such paths, and to $1$ if such paths exist but do not contain any fermionic edges.
\end{itemize}

Remember that the symmetric $n$ by $n$ matrix $X_{ab}(w_{\ell_B})$ is positive for any value of $w_\cJ$, hence the Gaussian measure $d\nu_{\cJ} $ is well-def\/ined. The matrix $Y_{\cB\cB'}(w_{\ell_F})$ is also positive, with all elements between 0 and 1. Since the slice assignments, the f\/ields, the measure and the integrand are now factorized over the connected components of $\bar \cJ$, the logarithm of $Z$ is exactly the same sum but restricted to the two-levels spanning trees:
\begin{gather} \label{treerep}
\log Z(\lambda, N) = \sum_{n=1}^\infty \frac{1}{n!} \sum_{\cJ \,{\rm tree}} \sum_J
 \int dw_\cJ \int d\nu_{ \cJ}  \partial_\cJ \prod_{\cB} \prod_{a\in \cB} \bigl[ W_{j_a} ( \sigma^a_{j_a} )
 \chi^{ \cB }_{j_a} \bar \chi^{\cB}_{j_a} \bigr],
\end{gather}
where the sum is the same but conditioned on $\bar \cJ = \cF_B \cup \cF_F$ being a \emph{spanning tree} on $V= [1, \dots , n]$. In~\cite{Gurau:2013oqa}, it is proven in detail that
\begin{thm} \label{thetheorem}
The series \eqref{treerep} is absolutely convergent for $\lambda\in [-1,1]$ uniformly in $j_{\max}$.
\end{thm}

\begin{thm} \label{thm:theorem2} The series \eqref{treerep} is absolutely convergent for $\lambda\in \mathbb{C}$, $\lambda = |\lambda|e^{\imath \gamma}$ in the domain $ |\lambda|^2 < (\cos2\gamma) $ uniformly in $j_{\max}$.
\end{thm}

We sketch below the proof of Theorem~\ref{thetheorem}, referring the reader to~\cite{Gurau:2013oqa} for details. By Cayley's theorem the number of two level trees over $n\ge 1$ vertices is exactly $2^{n-1}n^{n-2}$.

The Grassmann Gaussian integral evaluates to
\begin{gather*}
 \bigg( \prod_{\cB} \prod_{\genfrac{}{}{0pt}{}{a,b\in \cB}{a\neq b}} (1-\delta_{j_aj_b}) \bigg)
 \bigg( \prod_{\genfrac{}{}{0pt}{}{\ell_F \in \cF_F}{\ell_F=(a,b) } } \delta_{j_{a } j_{b } } \bigg) \Bigl( {\bf Y }^{\hat b_1 \dots \hat b_k}_{\hat a_1 \dots \hat a_k} +  {\bf Y }^{\hat a_1 \dots \hat b_k}_{\hat b_1 \dots \hat a_k}+\dots + {\bf Y }_{\hat b_1 \dots \hat b_k}^{\hat a_1 \dots \hat a_k} \Bigr) ,
\end{gather*}
where the sum runs over the $2^k$ ways to exchange an $a_i$ and a $b_i$. Each $ {\bf Y }^{\hat a_1 \dots \hat b_k}_{\hat b_1 \dots \hat a_k} $ factor is a determinant of a matrix made of $Y_{\cB\cB'}(w_{\ell_F})$ interpolating factors (see~\cite{Gurau:2013oqa} for the precise def\/inition). Its absolute value is therefore bounded by~1 thanks to Hadamard's inequality, because the corresponding matrix is positive with diagonal entries equal to~1.

The bosonic integral is a bit more cumbersome, as one should f\/irst evaluate the ef\/fect of the bosonic derivatives on the exponential vertex kernels $W_j$ through the Fa\`a di Bruno formula, whose combinatoric is easy to control. It leads to a sum over similar exponential kernels but multiplied by some polynomials.

To bound the remaining bosonic functional integral one f\/irst separates the exponential kernels from the polynomials by some Cauchy--Schwarz estimate with respect to the bosonic Gaussian measure. The exponential terms being positive, the corresponding piece is bounded by~1. The polynomial piece is then explicitly evaluated. This generates a dangerous product of local factorials of the number of f\/ields in the bosonic blocks, but allows also a good factor $M^{-j}$ from the propagator of scale~$j$ for each occupied bosonic scale~$j$.

But here comes the key point. The Grassmann Gaussian integrals ensure that the occupied scales in any bosonic block of the f\/irst forest formula are all \emph{distinct}. Therefore the good factor collected from the propagator easily beats the local factorials. The worst case is indeed when the~$p$ occupied scales in the block are lowest, in which case $\prod\limits_{j=1}^p M^{-j} = M^{-p(p+1)/2}$ which easily beats~$p!$.

For just renormalizable theories it is not so easy to beat the dangerous factors by the decay of the propagators, and the constructive expansion must proceed even more carefully, essentially expanding the functional integral in each scale in a much more detailed way.

\subsection{Iterated Cauchy--Schwarz bounds}

A Cauchy Schwarz bound requires to decompose an expression as a scalar product between two halves, $I= \langle S_1. S_2\rangle \implies \vert I \vert \le \langle S_1. S_1\rangle^{1/2} \langle S_2. S_2\rangle^{1/2}$. A key problem in the LVE is to bound the resolvents still present in the tree amplitudes by using their uniform bound in norm (in a~cardioid domain of the coupling constant). The idea of ICS bounds is to get rid of resolvents inductively, by applying many times this inequality. Typically at any inductive step a resolvent is sandwiched in the scalar product between $S_1$ and $S_2$ and disappears in the right hand side of the bound (thanks to the bound on its norm), resulting in tree amplitudes with at least one resolvent less.

ICS bounds for a general quartic tensor model were introduced in~\cite{Delepouve:2014bma}, but we describe here a~new slightly dif\/ferent version. More precisely in~\cite{Delepouve:2014bma} a tree amplitude of the LVE is decomposed in two halves by identifying a pair of opposite resolvents. However in this process the two halves may not have the same number of vertices, and the bound obtained can involve new trees of dif\/ferent orders. Here we shall describe how to work at f\/ixed order by pairing a resolvent with its opposite corner in the tree, so that the two halves $A$ and $B$ remain of the same order.

Consider a tensor model with quartic interactions for generalized colors $\cC \in \cD$. Let $N$ be the range for each color index. The perturbative amplitudes at order $n$, are simply $\lambda^n N^{F-(D-1)n}$, where $F$ is the number of \emph{faces} of the graph.

Among the graphs of the IFR are in particular the trees. A~planar colored tree of order $n$, is a tree joining $n+1$ loop vertices, in which each edge bears a~generalized color $\cC$, corresponding to the subset $\cC$ of indices of the quartic invariant for~$\cC$. Such a tree has therefore exactly $2n$ \emph{corners}, and an amplitude $N^{F-(d-1)n}$.

However in an LVE or MLVE a typical term is indexed by such planar colored trees but in which some corners of the tree can bear resolvents operators~$R$ instead of ordinary propaga\-tors~$1$. The goal of the ICBS bounds is to prove that in a cardioid domain of the coupling constant, the amplitude for any such tree is bounded (uniformly in~$\sigma$) by the supremum over all trees at the same order~$n$ of the same amplitude but of the ordinary perturbative type, that is without any resolvent.

A resolvent-dressed tree is a pair $(T,A, A^\dagger)$ made of a tree $T$ and a~subset~$S$ of the corners of~$T$. Its amplitude is
\begin{gather*} A_T^{A, A^\dagger} = \Tr \prod_{a \in A} C^{1/2} R C^{1/2}\prod_{a \in A^\dagger} C^{1/2} R^\dagger C^{1/2}\prod_{a \not\in A\cup A^\dagger} C
\prod_e \delta_{\cC_e},
\end{gather*}
where the product over $e$ runs over all edges of the tree, the edge factors are $\delta$ functions
\begin{gather*}
\delta_{\cC_e} = \prod_{c \in \cC_e} \delta_{i_c, j_c} \delta_{i'_c j'_c},
\end{gather*}
where the operator $ \delta_{i_c, j_c} \delta_{i'_c j'_c}$ joins the color indices $i_c$ and $j_c$ of the four corners touched by edge~$e$, and f\/inally the
trace means that indices have to be summed over all faces of the tree.

We assume that in a cardioid domain of the complex plane for the coupling constant a~uni\-form~$N$ and $\sigma$-independent bound for the resolvent holds:
\begin{gather*}
 \Vert R \Vert \le K .
\end{gather*}
The same independent bound automatically follows for the matrix coef\/f\/icients of the resolvent
\begin{gather*} \vert R_{n, \bar n} \vert \le K \qquad \forall\, n,\, \bar n,
\end{gather*}
since any coef\/f\/icient of a matrix is bounded in absolute value by its norm.

Iterated Cauchy--Schwarz bounds then allow to bound the absolute value of any resolvent-dressed tree amplitude by the supremum over similar tree amplitudes but without any resolvent for trees~$T'$ of the same order than~$T$:
\begin{thm} \label{theoICBS}
Under the hypotheses above, any resolvent-dressed tree $A_T^S$ of order $n$ obeys the bound
\begin{gather*}
 \big\vert A_T^{A, A^\dagger} \big\vert \le K^{2n} \sup_{T'} A_{T'}^\varnothing .
\end{gather*}
\end{thm}

\begin{proof} Consider a dressed tree $(T,A, A^\dagger)$ of order $n$ with $p = \vert A \vert + \vert A^\dagger \vert$ resolvents. If $p>0$ we select an arbitrary corner~$\gamma$ containing a resolvent $R$ of $T$ (if the resolvent is of the conjugate type $R^\dagger$ the reasoning is identical). Turning around the tree, there is a~unique corner $\tilde \gamma$ opposite to $\gamma$, that is such that the path from~$\gamma$ to $\tilde \gamma$ has $n-1$ corners
whether we turn around the tree clockwise or counterclockwise. If $p=0$, we select any corner and its opposite pair.

Having selected opposite corners we can divide the tree $T$ into two halves $S_1$ and $S_2$, made of the clockwise and counterclockwise paths form~$\gamma$
to $\tilde \gamma$. Every edge from $S_1$ to $S_2$ is a scalar product because of the form of the $\delta_e$ factors. Hence in a certain tensor product Hilbert space of many elements we can write $A_T^{A, A^\dagger} $ as
\begin{gather*}
\big\vert A_T^{A, A^\dagger}\big\vert= \vert \langle A_{S_1} R A_{S_2} \rangle \vert \le K \langle A_{S_1} A_{S_1}^\dagger \rangle^{1/2}\langle A_{S_2} A_{S_2}^\dagger \rangle^{1/2} .
\end{gather*}

In this way we have bounded the initial tree with $p$ resolvents by the arithmetic mean of two trees with total number of resolvents at most $2(p-1)$. Iterating this bound for all trees still containing resolvents, we obtain, after~$r$ iterations
\begin{gather*}
\big\vert A_T^{A, A^\dagger} \big\vert\le \prod K^{k_r} \prod_j \vert A_{S_j} \vert^{2^{-r}} .
\end{gather*}

It is not true that in such a process after a f\/inite number of steps all resolvents disappear (counterexamples are easily built). Nevertheless they \emph{rarefy} at each step, and this is enough to complete the proof of Theorem~\ref{theoICBS}. Let us def\/ine $p_r$ the total number of resolvents contained in all the trees $S_j$ and $m_r$ the number of trees containing at least one resolvent. Since any tree is of order $n$ it can contain at most $2n$ resolvents, hence $m_r \ge p_r/(2n) $. Furthermore $p_{r+1} \le 2p_r - 2m_r$ Therefore the sequence $q_r = 2^{-r}p_r$ satisf\/ies $q_{r+1} \le q_r (1-\frac{1}{2n})$ hence
tends to 0 as $r \to \infty$.

Now a rough relatively trivial bound to evaluate a tree with $p$ resolvents consists in paying a full additional factor $\prod_f N^{2q_f}$ for each face $f$ meeting $q_f$ resolvents along the face. Since a~resolvent is along at most $d$ dif\/ferent faces we certainly f\/ind that for any $r$
\begin{gather*}
 \big\vert A_T^{A, A^\dagger} \big\vert \le [f(N)]^{p} \sup_{T'} A_{T'}^\varnothing
\end{gather*}
for a certain function $f(N)$. Hence by the induction
\begin{gather*}
\big \vert A_T^{A, A^\dagger}\big\vert \le [f(N)]^{q_r} \sup_{T'} A_{T'}^\varnothing,
\end{gather*}
and letting $r \to \infty$ completes the proof of Theorem~\ref{theoICBS}, since $q_r \to 0$ as $r \to \infty$.
\end{proof}

\section{Results}

We now review the various tensor models for which (uniform) Borel summability has been proved. Rather than summarizing or paraphrasing the original papers, we focus on explaining the reader which of the previous techniques are required for which model and why.

\subsection{Melonic quartic models}

This is the simplest tensor model, restricted to a quartic melonic interactions, hence the set~$\mathcal{Q}$ of generalized colors in~\eqref{eq:model} is restricted to \emph{singletons}. This model is fully treated in~\cite{Gurau:2013pca} via a~single loop vertex expansion and does not require any ICS bounds. The LVE graphs in which all resolvents $R$ are replaced by one are exactly the tree graphs of the IFR of the model, and they coincide exactly with the melonic graphs of the initial representation of the tensor model. Hence the LVE automatically computes the leading order of quartic melonic tensor models as the $\sigma = 0$, $R(\sigma) =1$ approximation of the LVE. Sub leading orders are of course obtained by further expanding the resolvent factors of the LVE.

Why is it that the melonic quartic model does not require the technique of ICS bounds to prove Borel summability? This is essentially because in this case the propagators of the IFR correspond to insertions of the type $1 \otimes \cdots \otimes \sigma_c \otimes \cdots \otimes 1$ which all \emph{commute}.

As a result one can use the parametric representation of the resolvent
\begin{gather}\label{para} \big(1 + i \sqrt \lambda \sigma\big) = \int_0^\infty e^{-\alpha( 1 + i \sqrt \lambda \sigma) } d \alpha
\end{gather}
and regroup all identical colors. The bound now factorizes over the \emph{subforests of the same color} in the tree
and this allows for a constructive bound of the same order than the leading melonic graphs of the model.

\subsection{General quartic models}

For these models the interactions are no longer restricted to the melonic case. Intermediate f\/ields are matrices corresponding to generalized colors~$\cC$ which are not necessarily singletons. In particular they no longer commute when the subsets $\cC$ and $\cC'$ have a~non-trivial intersection $\cC \cap \cC' \ne \varnothing$. This is why the parametric representation \eqref{para} is not enough to get rid of resolvents in the loop vertex expansion, and ICS bounds have to be used~\cite{Delepouve:2014bma}.

\subsection[${\rm U}(1)-T^4_3$ TGFT]{$\boldsymbol{{\rm U}(1)-T^4_3}$ TGFT}

This f\/ield theory is built on the ${\rm U}(1)$ group, for tensors of rank~3, with quartic melonic interactions and a Boulatov-type projector, hence the propagator, in momentum space is
\begin{gather*} C(n) = \frac{\delta (n_1 + n_2 + n_3)}{n_1^2 + n_2^2 + n_3^2 +1}
\end{gather*}
for $n = (n_1, n_2 , n_3 )\in {\mathbb Z}^3$.

The delta projector reduces by 1 the ef\/fective dimension of the model, hence this model has the same power counting than a~rank~2 non-commutative f\/ield theory without any Boulatov projector, hence is fully convergent (no ultraviolet divergencies). Therefore it can be built by a~single loop vertex expansion~\cite{Lahoche:2015yya}.

\subsection[${\rm U}(1)-T^4_4$ TGFT]{$\boldsymbol{{\rm U}(1)-T^4_4}$ TGFT}

This f\/ield theory \cite{Lahoche:2015zya} is built on the ${\rm U}(1)$ group, for tensors of rank~4, with quartic melonic interactions and a~Boulatov-type projector, hence the propagator, in momentum space is
\begin{gather*} C(n) = \frac{\delta (n_1 + n_2 + n_3 + n_4)}{n_1^2 + n_2^2 + n_3^2 + n_4^2 +1}
\end{gather*}
for $n = (n_1, n_2 , n_3, n_4 ) \in {\mathbb Z}^4$.

The delta projector reduces by 1 the ef\/fective dimension of the model, hence this model has the same power counting than a rank~3 quartic tensor model without any Boulatov projector, hence requires a rather mild renormalization.

Therefore it requires a multi-loop vertex expansion. However it does not require any iterated Cauchy--Schwarz bounds, since insertions of dif\/ferent colors commute! This is surprising, given that a priori they do not commute with the~$C$ propagator, whose denominators mixes colors. However it is a peculiarity and simplifying feature that the Boulatov projector, when combined with quartic melonic interactions leads to a global momentum conservation rule along all the propagators of a given loop vertex.

Indeed consider a $\sigma$ insertion of color $c$ in a loop vertex. A priori the $n_c$ index may change into $m_c$. However the propagators $\delta$ functions immediately before and after the insertion enforce that $n_c = - \sum\limits_{c' \ne c} n_{c'}$. Therefore since the $n_{c'}$ are conserved through the~$\sigma_c$ insertions, we obtain that $n_c = m_c$. In fact the random intermediate f\/ield matrices $\sigma_c$ reduce to their diagonal part in momentum space, and the intermediate f\/ields can therefore be considered as vectors rather than matrices~\cite{Lahoche:2015ola}.

Now this conservation rule along any loop vertex ensures the commutation of all intermediate f\/ields along the loop vertex, hence the argument with the parametric representation of~\cite{Gurau:2013pca} can be adapted to this situation and ICS bounds are not necessary.

\subsection[$T^4_3$ TFT]{$\boldsymbol{T^4_3}$ TFT}

This super-renormalizable f\/ield theory is built on the ${\rm U}(1)$ group, for tensors of rank~3, with quartic melonic interactions and an ordinary Laplacian-based projector which in momentum space is
\begin{gather*} C(n) = \frac{1}{n_1^2 + n_2^2 + n_3^2 +1}
\end{gather*}
for $n = (n_1, n_2 , n_3 ) \in {\mathbb Z}^3$.

The model has a power counting almost similar to the one of the ordinary $\phi^4_2$ theory. It has for each color $c$ two vacuum divergent graphs, one linearly divergent and one logarithmically divergent. It has also a single logarithmically divergent two-point graph (the ``melonic tadpole''), again with a single vertex, which requires a mass renormalization (see Fig.~\ref{divergences}).

\begin{figure}[t]\centering
\includegraphics[width=0.55\textwidth]{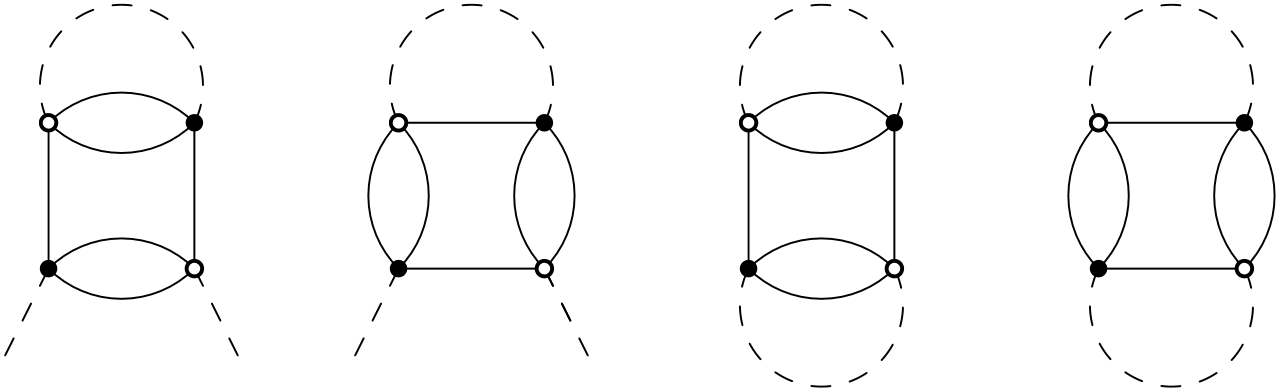}
 \caption{From left to right, the divergent tadpole, the convergent tadpole and the two vacuum divergent graphs.} \label{divergences}
\end{figure}

Therefore the proof of Borel summability for this model requires both a multi-loop vertex expansion and ICS bounds \cite{Delepouve:2014hfa}. Indeed insertions of intermediate f\/ields of dif\/ferent colors along a loop vertex no longer need to commute, since they do not commute with the propagators along a given loop vertex. Indeed these propagators may have now dif\/ferent momenta running through them, since there is no longer the conservation rule of the quartic melonic models with Boulatov projectors.

An interesting aspect of the renormalization of this model is that it is performed directly in the IFR representation. Hence the initial $\Tr \log (1 + i \sqrt \lambda \Sigma )$ interaction is replaced by a subtracted interaction $\Tr \log_2 (1 + i \sqrt \lambda \Sigma )$, where $\log_2 (1 + x ) := -x + \log (1 + x )$. In addition to the ICS bounds to get rid of resolvents in the trees created by the bosonic forest formula, there is also a~non-trivial argument to bound the non perturbative part left in the exponential by the MLVE. This non trivial ``non-perturbative'' bound can be traced to the divergence of the vacuum energy graphs, and can be considered a kind of tensor analog in the IFR representation of the famous Nelson bound for~$\phi^4_2$, which was the birth act of constructive f\/ield theory~\cite{nelson}.

\subsection{Quartic matrix model}

The paper \cite{Gurau:2014lua} deals with a Gaussian complex random matrix model with a quartic perturbation. Its partition function is def\/ined by the integral over complex $N\times N$ matrices $M$:
\begin{gather*}
 Z(\lambda,N)= \int dM e^{-\Tr( MM^{\dagger}-\frac{\lambda}{2N}MM^{\dagger}MM^{\dagger})} .
\end{gather*}

The model can be shown Borel summable with a single loop vertex expansion but the trees of that expansion form only a~part of the planar sector of the theory which is the leading order at large $N$. To extract this leading sector is possible through an additional expansion called the topological expansion in~\cite{Gurau:2014lua}. Starting form the LVE trees, one can perform a f\/irst expansion step on the resolvents, testing the presence of an additional f\/ield $\sigma$. This f\/ield is then contracted with the Gaussian measure, and the process is continued until a f\/irst non-planar Wick contraction appear. The part of the expansion for planar graphs does not diverge since the number of planar graphs at order~$n$ is bounded by $({\rm const})^n$.

Continuing this additional topological expansion until at most $p+1$ independent crossing, it is possible to extract the complete terms of the~$1/N$ expansion up to order $p$ and to prove that the remainder is still Borel summable, with a uniform bound in $N^{- (p+1)}$.

A similar expansion could be used for tensor models such as those of~\cite{Bonzom:2015axa}, which mix a melonic sector and a planar sector behavior.

\section{Conclusion: open questions}

Super-renormalizable and just renormalizable tensor f\/ield theories form a rich world \cite{Geloun:2013saa}. A~key physical issue is to move towards realistic quantum gravity models, and to better connect the ``tensor f\/ield theory'' approach to other approaches to quantum gravity. In parti\-cu\-lar understanding which tensor models admit a realistic emergent four dimensional (Euclidean) space time with a real time analytic continuation seems to us the key problem. However we shall not discuss this issue here, addressing the reader to our twin review~\cite{Rivasseau:2016zco} and focus here only on a brief discussion of the next technical steps of the constructive tensor program.

The next obvious step in this program is to build super-renormalizable theories with power counting of the same level than ordinary $\phi^4_3$, namely the $T^4_4$ model with propagator $(p^2 + 1)^{-1}$. This is well under way and it seems that it can be treated like~$T^4_3$ by just the same 2-jungle formula of the MLVE and ICS bounds; but the list of divergent graphs is much longer (there are 12 primitively divergent graphs) and the $\Tr \log$ interaction should therefore be pushed further, resulting in particular in much more complicated non perturbative estimates compared to~$T^4_3$~\cite{RVT}.

The model ${\rm U}(1)-T^4_5$ with Boulatov projectors has a comparable f\/inite amount of divergent graphs. It should be easier than $T^4_4$, since intermediate f\/ields can be represented by diagonal matrices in the momentum basis.

The next, more dif\/f\/icult step is the construction of a just renormalizable model such as the $T^4_5$ model or the ${\rm U}(1)-T^4_6$ model with Boulatov projector \cite{Lahoche:2015ola, Samary:2014oya,Samary:2012bw}. Both models are asymptotically free \cite{Geloun:2013saa,Rivasseau:2015ova,Samary:2013xla}, hence the construction should be possible but it is not clear whether it can be done with the same tools or if a more advanced expansion (such as a multiscale loop vertex expansion with higher-order jungle formulas) will be required.

Another important direction is to treat models with higher-order interactions, for instance the just renormalizable models with six-order melonic interactions introduced in \cite{BenGeloun:2011rc} and \cite{Carrozza:2013wda}. Here the starting point might be the intermediate f\/ield representation introduced in~\cite{Rivasseau:2010ke} and improved in \cite{LionniRivasseau}. Beware however that even in zero dimension we do not know yet how to perform a loop vertex expansion in this representation \cite{LionniRivasseau}.

We would like to complete this review by indicating two other challenging research directions for matrix models and non-commutative f\/ield theory

\begin{itemize}\itemsep=0pt

\item The f\/irst direction is to extend the analyticity domain in the coupling constant of quartic models beyond the cardioid characteristically obtained through the LVE. In the case of vector models we already know that this is possible and that the optimal angular domain of analyticity extends to~3$\pi$ (see Fig.~\ref{extendedcardioidfig} and the last section of~\cite{Gurau:2014vwa}). In matrix or tensor models nothing of this type is known.

This is a key issue for many reasons, in particular to understand the instanton cuts which are responsible for the usual non-perturbative ef\/fects of $\phi^4$ theory, and how they should disappear in the limit $N \to \infty$, unveiling a new singularity at f\/inite distance on the negative coupling constant which is the one responsible for the single and double scaling limits of quantum gravity in two dimensions, hence for the emergence of continuous surfaces~\cite{DiFrancesco:1993cyw}.

\item Construct non-perturbatively the full non-planar sector of the Grosse--Wulkenhaar mo\-del~\cite{Grosse:2004yu}. Although in a certain limit $\theta \to \infty$ of inf\/inite non-commutativity, this sector disappears, again it would be very interesting to understand rigorously how to construct the full theory at f\/inite $\theta$. It would be extremely interesting, as the theory is asymptotically safe, to check that this property persists at the non-perturbative level.

\end{itemize}

\subsection*{Acknowledgements}
We thank A.~Abdesselam, T.~Delepouve, R.~Gurau, V.~Lahoche, L.~Lionni, J.~Magnen, K.~Noui, M.~Smerlak, A.~Tanasa, F.~Vignes-Tourneret and ZhiTuo Wang for contributing to various aspects of this constructive program.

\pdfbookmark[1]{References}{ref}
\LastPageEnding

\end{document}